\documentclass[aos,preprint]{imsart}

\RequirePackage[OT1]{fontenc}
\RequirePackage{amsthm,amsmath}
\RequirePackage[sort&compress, authoryear]{natbib}
\RequirePackage[colorlinks,citecolor=blue,urlcolor=blue]{hyperref}

\startlocaldefs
\numberwithin{equation}{section}
\theoremstyle{plain}
\newtheorem{thm}{Theorem}[section]
\endlocaldefs

\usepackage{amssymb}
\usepackage{graphicx}
\usepackage{multirow}
\usepackage{array}
\usepackage{epstopdf}

\newtheorem{lem}{\protect\lemmaname}
\providecommand{\lemmaname}{Lemma}

\DeclareMathOperator*{\argmin}{arg\,min}
\begin{document}

\begin{frontmatter}
\title{Optimal Estimation for the Functional Cox Model\thanksref{T1}}
\thankstext{T1}{The research of Jane-Ling Wang is supported by NSF grant DMS-0906813. The research of Xiao Wang is supported by NSF grants CMMI-1030246 and DMS-1042967.}

\begin{aug}
\author{\fnms{Simeng} \snm{Qu}\thanksref{m1}\ead[label=e1]{qu20@purdue.edu}},
\author{\fnms{Jane-Ling} \snm{Wang}\thanksref{m2}\ead[label=e2]{janelwang@ucdavis.edu}}
\and
\author{\fnms{Xiao} \snm{Wang}\thanksref{m1}
\ead[label=e3]{wangxiao@purdue.edu}
}


\affiliation{Department of Statistics, Purdue University\thanksmark{m1} \\
   Department of Statistics, University of California at Davis\thanksmark{m2}}

\address{S. Qu\\
X. Wang\\
Department of Statistics\\
Purdue University\\
West Lafayette, IN 47907\\
USA\\
\printead{e1}\\
\phantom{E-mail:\ }\printead*{e3}}

\address{J.-L. Wang\\
Department of Statistics\\
University of California\\
Davis, CA 95616\\
USA\\
\printead{e2}}

\end{aug}

\begin{abstract}
Functional covariates are common in many medical, biodemographic, and neuroimaging studies. The aim of this paper is to study functional Cox models with right-censored data in the presence of both functional and scalar covariates. We study the asymptotic properties of the maximum partial likelihood estimator and establish the asymptotic normality and efficiency of the  estimator of the finite-dimensional estimator. Under the framework of  reproducing kernel Hilbert space, the estimator of the coefficient function for a functional covariate achieves the minimax optimal rate of convergence under a weighted $L_2$-risk. This optimal rate is determined jointly by the censoring scheme, the reproducing kernel and the covariance kernel of the functional covariates. Implementation of the estimation approach and the selection of the smoothing parameter are discussed in detail. The finite sample performance is illustrated by simulated examples and a real application.
\end{abstract}



\end{frontmatter}
\section{Introduction}

The proportional hazard model, known as the Cox model, was introduced by \cite{Cox1972}, where the  hazard function of the survival time $T$ for a subject with  covariate $Z(t) \in \mathbb R^p$ is represented by
\begin{equation}\label{eq:cox}
h(t|Z) = h_0(t) e^{\theta_0' Z (t)},
\end{equation}
where $h_0$ is an unspecified baseline hazard function and  $\theta_0\in \mathbb R^p$  is an  unknown parameter.  Some or all of the $p$ components in $Z$ may be time-independent, meaning that they are constant over time $t$, or may depend on $t$.  The aim of this paper is to develop a different type of model,  the  functional Cox model, by incorporating functional predictors along with  scalar predictors. \cite{Wang2011} first proposed such a model when studying  the survival of diffuse large-B-cell lymphoma (DLBCL) patients, which is thought
to be influenced by genetic differences. The functional predictor, denoted by $X(\cdot): {\cal S}\rightarrow \mathbb R$ on a compact domain ${\cal S}$, is a smooth stochastic process related to the high-dimensional microarray gene expression of DLBCL patients. The entire trajectory of $X$ has an effect on the hazard function, which makes it 
different from the Cox model (\ref{eq:cox}) with time-varying covariates, where only the current value of $X$ at time $t$ affects the hazard function  at time $t$. Specifically,
the functional Cox model with a vector covariate $Z$ and functional covariate $X(t)$ represents the hazard function by
\begin{equation}
h(t\,|X)=h_{0}(t)\exp\Big\{\theta_0'Z+\int_{\mathcal{\mathcal{S}}}X(s)\beta_0(s)ds\Big\},\label{eq:model}
\end{equation} 
where $\beta_0$ is an unknown coefficient function. Without loss of generality, we take $\mathcal{S}$ to be $[0,1]$. 


Under the right censorship model and letting  $T^{u}$ and $T^{c}$ be, respectively, the failure time and censoring
time,
we observe i.i.d. copies of $(T,\,\Delta,\, X(s), s\in {\cal S})$, $(T_{1},\,\Delta_{1},\, X_{1}),\ldots, (T_{n},\,\Delta_{n},\, X_{n})$, where 
 $T=\min\{T^{u},\, T^{c}\}$ is the observed time event and $\Delta=I\{T^{u}\leq T^{c}\}$ is the censoring indicator. 
Our goal is to estimate $\alpha_0 = (\theta_0, \beta_0(\cdot))$  to reveal how the functional covariates $X(\cdot)$ and other scalar covariates $Z$ relate to survival.

Let $\hat\alpha = (\hat\theta, \hat\beta(\cdot))$ be an estimate from the data. It is critical to define the risk function to measure the accuracy of the estimate.
\  Let $W = (Z, X)$ and
\[
\eta_\alpha(W) = \theta' Z + \int_0^1 \beta(s) X(s)ds.
\]
Define an $L_2$-distance such that
\begin{equation}\label{equ:d}
d^2(\hat\alpha, \alpha_0) = \mathbb E \Big\{\Delta \Big( \eta_{\hat\alpha}(W) - \eta_{\alpha_0}(W) \Big)^2\Big\}.
\end{equation}
Based on this $L_2$-distance, we show that the accuracy of $\hat\theta$ is measured by the usual $L_2$-norm $\|\hat\theta - \theta\|_2$ and the accuracy of $\hat\beta$ is measured by a weighted $L_2$-norm $||\hat{\beta}-\beta_0||_{C_\Delta}$, where
\[
C_{\Delta}(s,t)=\mathrm{Cov}\Big(\Delta X(s), ~\Delta X(t)\Big), ~~~ \mbox{~and~}~~~ \|\beta\|_{C_{\Delta}}^{2}=\int\int\beta(s)C_{\Delta}(s,t)\beta(t)dsdt.
\]
We now explain why we do not consider the convergence of $\hat\beta$ with respect to the usual $L_2$-norm in the present paper. In general, $\|\hat\beta - \beta_0\|_2^2=\int_0^1(\hat\beta(t) - \beta_0(t))^2dt$ may not converge to zero in probability, and to obtain the convergence of $\|\hat\beta - \beta_0\|_2^2$ one needs additional smoothness conditions linking  $\beta$ to the functional predictor $X$; see \cite{Crambes2009} for a discussion of this phenomenon for functional linear models. 
On the other hand, in the presence of censoring, the Kullback-Leibler distance between two probability measures $\mathbb P_{h_0,\hat\alpha}$ and $\mathbb P_{h_0,\alpha_0}$ is equivalent to the $L_2$ distance $d$ in (\ref{equ:d}). When failure times $T^u$ are  fully observed, i.e. $\Delta=1$ is true regardless of $X(s)$, the $\|\cdot\|_{C_\Delta}$ norm becomes $\|\cdot\|_{C}$, where $C(t, s) = \mathrm{Cov}(X(t), X(s))$ is the covariance function of $X$.
This norm $\|\cdot\|_{C}$ has been widely used for functional linear models (e.g. \citealp{CaiYuan2012}).

Many people have studied parametric, nonparametric, or semiparametric modeling of the covariate effects using the Cox model (e.g. \citealp{Sasieni1992a, Sasieni1992b}; \citealp{Hastie1986,Hastie1990}; \citealp{Huang1999} and references therein) and 
\cite{Cox1972}  proposed to use partial likelihood to estimate $\theta$ in (\ref{eq:cox}). The advantage of using  partial likelihood is that it estimates $\theta$ without knowing or involving  the functional form of $h_0$.
The asymptotic equivalence of the partial likelihood estimator and the maximum likelihood estimator has been established by several authors (\citealp{Cox1975}; \citealp{Tsiatis1981}; \citealp{Andersen1982}; \citealp{Johansen1983}; \citealp{Jacobsen1984}). 
On the other hand, the literature on functional regression, in particular for functional linear models, is too vast to be summarized here. Hence, we only refer to the well-known monographs \cite{Ramsay2005} and \cite{Ferraty2006}, and some recent developments such as \cite{James2002}, \cite{Muller2005}, \cite{Hall2007}, \cite{Crambes2009}, \cite{YuanCai2010}, \cite{CaiYuan2012} for further references.  Recently, {\cite{Kong2014} studied a similar functional Cox model to establish some asymptotic properties but without investigating the optimality property. Moreover, their estimate of the parametric component converges at a rate which is slower than root-n. Thus, it is desirable to develop new theory to systematically investigate properties of the estimates and establish their  optimal asymptotic properties. In addition, instead of assuming that both $\beta_0$ and $X$ can be represented by the same set of basis functions, we adopt a more general reproducing kernel Hilbert space framework to estimate the coefficient function.  

In this paper we study the  convergence of   the estimator  $\hat{\alpha}= (\hat{\theta},  \hat{\beta})$  under the framework of the reproducing kernel Hilbert space and the Cox model. The  true coefficient function $\beta_0$ is assumed to reside in a reproducing
kernel Hilbert space $\mathcal{H}(K)$ with the reproducing kernel $K$, which is
a subspace of the collection of square integrable functions on $[0, 1]$. There are two main challenges for our asymptotic analysis,  the nonlinear structure of the Cox model, and the fact that  the reproducing kernel $K$ and the covariance kernel $C_\Delta$ may not share a common ordered set of eigenfunctions, so $\beta_0$ can not be represented effectively by the leading eigenfunctions of $C_\Delta$.
We obtain the estimator by maximizing a penalized partial likelihood and establish   $\sqrt{n}$-consistency, asymptotic normality, and semi-parametric efficiency  of the estimator  $\hat{\theta}$ of the finite-dimensional regression parameter.   

A second optimality result is on the estimator of the coefficient function, which achieves the minimax optimal rate of convergence under the weighted $L_2$-risk.
The optimal rate of convergence is established in the following two steps. 
First, the convergence rate of the penalized partial likelihood estimator is calculated. Second, in the presence of the nuisance parameter $h_0$, the minimax lower bound on the risk is derived, which matches the  convergence rate of the partial likelihood estimator.  Therefore the estimator is rate-optimal. 
Furthermore, an efficient algorithm is developed to estimate the coefficient function. Implementation of the estimation approach, selection of the smoothing parameter, as well as calculation of the information bound $I(\theta)$ are all discussed in detail.

The rest of the paper is organized as follows. Section 2 summarizes the main results regarding the asymptotic analysis of the penalized partial likelihood predictor. Implementation of the estimation approach is discussed in Section 3, including a GCV method to select the  smoothing parameter and a method of calculating the information bound of $\theta$ based on the alternating conditional expectations (ACE) algorithm.  Section 4 contains numerical studies, including simulations and a data application. All the proofs are relagated to Section 5, followed by several technical details in the Appendix.

\section{Main Results}\label{sec:main}


We estimate $\alpha_{0}=(\theta_0, \beta_0)\in {\mathbb R}^p\times {\cal H}(K)$ by maximizing the penalized log partial likelihood,  
\begin{equation}
\hat{\alpha}_{\lambda}={\arg\min}_{\alpha\in{\mathbb R}^p\times\mathcal{H}(K)}l_{n}(\alpha)+\lambda~ J(\beta), \label{eq:penal lh}
\end{equation}
 where the negative log partial likelihood is given by
\begin{equation}\label{eq:log partial lh}
l_{n}(\alpha)=-{\frac{1}{n}}\sum_{i=1}^{n}\Delta_{i}\Big\{\eta_\alpha(W_i)-\log\sum_{T_{j}\geq T_{i}}\exp(\eta_\alpha(W_j))\Big\},
\end{equation}
$J$ is a penalty function controlling the smoothness
of $\beta$, and $\lambda$ is a smoothing parameter that balances
the fidelity to the model and the plausibility of $\beta$. 
The choice of the penalty function $J(\cdot)$ is a squared
semi-norm associated with $\mathcal{H}$ and its norm. In general, $\mathcal{H}(K)$ can be decomposed with respect to the penalty $J$ as
$\mathcal{H}=\mathcal{N}_{J}+\mathcal{H}_{1}$, where $\mathcal{N}_{J}$
is the null space defined as
\[
\mathcal{N}_{J}=\{\beta\in\mathcal{H}(K):\, J(\beta)=0\},
\]
 and $\mathcal{H}_{1}$ is its orthogonal complement in $\mathcal{H}$.
Correspondingly, the kernel $K$ can be decomposed as $K=K_{0}+K_{1}$, where
$K_{0}$ and $K_{1}$ are kernels for the subspace $\mathcal{N}_{J}$
and $\mathcal{H}_{1}$ respectively.  For example,
for the Sobolev space, 
\[
\mathcal{W}_{2,m}=\Big\{f\,:[0,1]\rightarrow R |  \ f,f',\ldots f^{(m-1)} \text{are absolutely continuous},\, f^{(m)}\in L_2\Big\},
\]
 endowed with the norm
\begin{equation}
||f||_{\mathcal{W}_{2,m}}=\sum_{v=0}^{m-1}f^{(v)}(0)+\int_{0}^{1}(f^{(m)}(s))^{2}ds,\label{eq:Wm norm}
\end{equation}
where the penalty $J(\cdot)$ in this case can be assigned as
$J(f)=\int_{0}^{1}(f^{(m)}(s))^{2}ds$.

\medskip


We first present some main assumptions: 
\begin{description}
\item [{(A1)}] Assume $\mathbb{E}(\Delta Z)=0$ and $\mathbb E(\Delta X(s))=0, s\in [0,1]$.
\item [{(A2)}] The failure time $T^{u}$ and the censoring time $T^{c}$
are conditionally independent given $W$.
\item [{(A3)}]  The observed event time $T_{i},\ 1\leq i\leq n$ is
in a finite interval, say $[0,\tau]$, and there exists a small positive constant $\varepsilon$ such
that:  (i) $\mathbb P(\Delta=1|W)>\varepsilon$,  and (ii) $\mathbb P(T^{c}>\tau|W)>\varepsilon$
almost surely with respect to the probability measure of $W$.

 \item [{(A4)}]  The covariate $Z$ takes values in a bounded subset of $\mathbb{R}^p$, and  the $L_2$-norm $||X||_{2}$ of $X$ is bounded almost surely.

\item[{(A5)}] Let $0<c_1<c_2<\infty$ be two constants. The baseline joint density $f(t,\Delta=1)$ of $(T,\Delta=1)$ satisfies $ c_1<f(t,\Delta=1)<c_2$ for all $t\in [0,\tau]$.
\end{description}

Condition (A1) requires $Z$ and $X$ to be suitably centered.  Since  the partial likelihood function (\ref{eq:log partial lh})
does not change when centering  $Z_{i}$ as $Z_{i}-\sum\Delta_i Z_{i}/\sum\Delta_i$ or $X_{i}$ as $X_{i}-\sum\Delta_i X_{i}/\sum\Delta_i$,
centering does not impose any real restrictions.
In addition, centering by $\mathbb E(\Delta Z)$ and $\mathbb E(\Delta X)$, instead of centering by $\mathbb E(Z)$ and  $\mathbb E(X)$,
simplifies the asymptotic analysis.
Conditions (A2) and (A3) are common assumptions for analyzing right-censored
data, where (A2) guarantees the censoring mechanism to be non-informative
while (A3) avoids the
unboundedness of the partial likelihood at the end point of the support
of the observed event time. This is a reasonable assumption since the
experiment can only last for a certain amount of time in practice. 
Assumption (A3)(i) further ensures
the probability of being uncensored to be positive regardless of the
covariate and (A3)(ii) controls the censoring  rate so that it will
not be too heavy. 
Assumption (A4) places a
boundedness restriction on the covariates. This assumption can be relaxed to the sub-Gaussianity of $||X||_{2}$, which implies that with a large probability, $||X||_{2}$ is bounded. Condition (A5) and condition (A1) together guarantee the identifiability of the model. Moreover the joint density $f(T,Z,X,\Delta=1)$ 
is bounded away from zero and infinity under assumptions (A3)-(A5), which is used to calculate the information bound and convergence rate later in Theorem \ref{thm:score function} and Theorem \ref{thm: convergence rate}.

Let $r(W)=\exp(\eta_{\alpha}(W))$, then the counting process martingale
associated with model (1) is:
\[
M(t)=M(t|W)=\Delta I\{T\leq t\}-\int_{0}^{t}I\{T\geq u\}r(W)dH_{0}(u),
\]
where $H_{0}(t)=\int_0^t h_0(u)du$ is the baseline cumulative hazard function. For two sequences ${a_{k}:k\ge1}$ and ${b_{k}:k\ge1}$
of positive real numbers, we write $a_{k}\asymp b_{k}$ if there are
positive constants $c$ and $C$ independent of $k$ such that $c\le a_{k}/b_{k}\le C$
for all $k\ge1$.

\begin{thm}\label{thm:score function}
Under (A1)-(A5), the
efficient score for the estimation of $\theta$ is 
\[
l_{\theta}^{*}(T,\Delta,W)=\int_{0}^{\tau}(Z-a^{*}(t)-\eta_{g^{*}}(X))dM(t)
\]
where $(a^{*},\, g^{*})\in{L}_{2}\times\mathcal{H}(K)$
is a solution that minimizes
\[
\mathbb{E}\Big\{\Delta\|Z-a(T)-\eta_{g}(X)\|^{2}\Big\}.
\]
Here $a^{*}$ can be expressed as $a^{*}(t)=\mathbb{E}[Z-\eta_{g^{*}}(X)|\, T=t,\,\Delta=1]$.
The information bound for the estimation of $\theta$ is
\[
I(\theta)=\mathbb{E}[l_{\theta}^{*}(T,\Delta,W)]^{\otimes2}=\mathbb{E}\{\Delta[Z-a^{*}(T)-\eta_{g^{*}}(X)]^{\otimes2}\},
\]
where $y^{\otimes2}=yy'$ for column vector $y\in \mathbb R^{d}$.
\end{thm}

Recall that $K$ and $C_{\Delta}$ are two real, symmetric, and nonnegative definite functions. Define a new kernel $K^{1/2}C_{\Delta}K^{1/2}: [0,1]^{2}\rightarrow\mathbb{R}$, which is a real, symmetric,  square integrable, and nonnegative definite
function. Let $L_{K^{1/2}C_{\Delta}K^{1/2}}$ be the corresponding linear operator
${L}_{2}\rightarrow{L}_{2}$. Then Mercer’s theorem (\citealp{Riesz1955})
implies that there exists a set of orthonomal eigenfunctions $\{\phi_{k}:k\geq1\}$
and a sequence of eigenvalues $s_{1}\geq s_{2}\geq\ldots>0$ such
that 
\[
K^{1/2}C_{\Delta}K^{1/2}(s,t)=\sum_{k=1}^\infty s_{k}\phi_{k}(s)\phi_{k}(t), ~~~~~L_{K^{1/2}C_{\Delta}K^{1/2}} (\phi_k) = s_k.  
\]

\begin{thm}\label{thm: convergence rate} Assume (A1)-(A5) hold. 
\begin{enumerate}
\item[(i)](consistency) $d(\hat{\alpha},\alpha_{0})\stackrel{p}{\rightarrow}0$,
provided that $\lambda\rightarrow0$ as $n\rightarrow\infty$. 
\item[(ii)](convergence rate) If the eigenvalues $\{s_{k}:k\geq1\}$ of
$K^{1/2}C_{\Delta}K^{1/2}$ satisfy $s_{k}\asymp k^{-2r}$ for some
constant $0<r<\infty$,  then for  $\lambda=O(n^{-\frac{2r}{2r+1}})$ we have  \[
d(\hat{\alpha},\alpha_{0})=O_p(n^{-\frac{r}{2r+1}}). 
\]

\item[(iii)] If $I(\theta)$ is nonsingular, then 
$\|\hat{\theta}-\theta_{0}\|_{2}=O_p(n^{-\frac{r}{2r+1}})$
and 
\[
\lim_{A\rightarrow\infty}\lim_{n\rightarrow\infty}\sup_{\beta_{0}\in\mathcal{H}(K)}\mathbb{P}_{h_{0}\beta_{0}}\Big\{\|\hat{\beta}_{\lambda}-\beta_{0}\|_{C_{\Delta}}\geq An^{-\frac{r}{2r+1}}\Big\}=0.
\]
\end{enumerate}
\end{thm}


Theorem \ref{thm: convergence rate} indicates that the convergence rate is determined by the
decay rate of the eigenvalues of $K^{1/2}C_{\Delta}K^{1/2}$, which is
jointly determined by the eigenvalues of both reproducing kernel $K$
and the conditional covariance function $C_{\Delta}$ as well as  by the
alignment between $K$ and $C_{\Delta}$. When $K$ and $C_{\Delta}$ are  perfectly aligned, meaning that  $K$ and $C_{\Delta}$ have the same ordered eigenfunctions,
the decay rate of $\{s_{k}:k\geq1\}$  equals to the summation
of the decay rates of the eigenvalues of $K$ and $C_{\Delta}$. \cite{CaiYuan2012} established a similar result for functional linear models, for which  the optimal prediction risk depends on the decay rate of the eigenvalues of $K^{1/2}CK^{1/2}$, where  $C$ is the covariance function of $X$.  

The next theorem establishes the asymptotic normality of $\hat\theta$ with root-n consistency.
\begin{thm}
\label{thm: asymp normal}Suppose (A1)-(A5) hold, and that  the Fisher information 
$I(\theta_0)$ is nonsingular. Let $\hat{\alpha}=(\hat{\theta},\hat{\beta})$
be the estimator given by (\ref{eq:penal lh}) with $\lambda=O(n^{-\frac{2r}{2r+1}})$.
Then
\[
\sqrt{n}(\hat{\theta}-\theta_0)=n^{-1/2}I^{-1}(\theta_0)\sum_{i=1}^{n}l_{\theta_0}^{*}(T_{i},\Delta_{i},W_{i})+o_{p}(1)\stackrel{d}{\rightarrow}\mathcal{N}(0,\Sigma), 
\]
where $\Sigma=I^{-1}(\theta_0)$.
\end{thm}
For the nonparametric coefficient function $\beta$, it is of interest to see whether the convergence rate of $\hat{\beta}$ in Theorem \ref{thm: convergence rate} is optimal. In the following, we derive a minimax lower bound for the risk.

\begin{thm}\label{thm: lower bound} 
Assume that the baseline hazard function $h_0 \in {\cal F} = \{h: H(t)=\int_0^t h(s)ds<\infty,\text{ for any }0<t<\infty\}$.
Suppose that the eigenvalues $\{s_{k}:k\geq1\}$ of $K^{1/2}C_{\Delta}K^{1/2}$
satisfy $s_{k}\asymp k^{-2r}$ for some constant $0<r<\infty$.  Then,
\[
\lim_{a\rightarrow0}\lim_{n\rightarrow\infty}\inf_{\hat{\alpha}}\sup_{\alpha_{0}\in\mathbb R^p\times\mathcal{H}(K)} \sup_{h_0\in {\cal F}}\mathbb P_{\alpha_0, h_0}\Big\{\big\|\hat\beta - \beta_0\big\|_{C_\Delta}\geq an^{-\frac{r}{2r+1}}\Big\}=1,
\]
where the infimum is taken over all possible predictors $\hat{\alpha}$
based on the observed data.
\end{thm}

Theorem \ref{thm: lower bound} shows that the minimax lower bound of the convergence rate
for estimating $\beta_{0}$ is $n^{-r/(2r+1)}$, which is determined
by $r$ and the decay rate of the  eigenvalues of $K^{1/2}C_{\Delta}K^{1/2}$. 
We have shown that this rate is achieved by the penalized partial likelihood predictor
 and therefore this estimator is rate-optimal. 


\section{Computation of the Estimator}

\subsection{Penalized partial likelihood}

In this section, we present an algorithm to compute the penalized partial likelihood estimator.
Let $\{\xi_{1},\ldots\xi_{m}\}$
be a set of orthonormal basis of the null space with $m=\dim(\mathcal{N}_{J})$.
The next theorem provides a closed form representation of $\hat{\beta}$ from the penalized partial likelihood method.

\begin{thm} \label{thm: beta closed form} 
The penalized partial likelihood estimator of the coefficient function is given by
\begin{equation}
\hat{\beta}_{\lambda}(t)=\sum_{k=1}^{m}d_{k}\xi_{k}(t)+\sum_{i=1}^{n}c_{i}\int_{0}^{1}X_{i}(s)K_{1}(s,t)ds, \label{eq: beta closed form}
\end{equation}
where $d_{k}\ (k=1,\ldots m)$ and $c_{i} \ (i=1,\ldots n)$ are constant coefficients.
\end{thm}

Theorem $\ref{thm: beta closed form}$ is a direct application of the generalized
version of the well-known representer lemma for smoothing splines
(see  \cite{Wahba1990} and \cite{YuanCai2010}). We omit the proof here. In fact, the algorithm can be made more efficient without using all $n$ bases $\int_{0}^{1}X_{i}(s)K_{1}(s,t)ds$, $i=1, \ldots, n$ in (\ref{eq: beta closed form}). \cite{Gu2013} showed that, under some conditions, a more efficient estimator, denoted by $\beta_{\lambda}^{*}$, sharing the same convergence rate with $\hat{\beta}_{\lambda}$,
can be calculated in the data-adaptive finite-dimensional space 
\[
\mathcal{H}^{*}=\mathcal{N}_{J}\oplus\big\{ K_{1}(\tilde{X}_{j},\cdot),\, j=1,\ldots,q\big\},
\]
where $\{\tilde{X}_{j}\}$ is a random subset of $\{X_{i}:\Delta_{i}=1\}$
and 
\[
K_{1}(\tilde{X}_{j},\cdot)=\int_{0}^{1}\tilde{X}_{j}(s)K_{1}(s,\cdot)ds.
\]
Here,
$q=q_{n}\asymp n^{2/(ps+1)+\epsilon}$ for some $s>1$ and     $ \ p\in[1,2]$, 
and for any $\epsilon>0$. Therefore, $\beta_{\lambda}^{*}$
is given by
\[
\beta_{\lambda}^{*}(t)=\sum_{k=1}^{m}d_{k}\xi_{k}(t)+\sum_{j=1}^{q}c_{j}K_{1}(\tilde{X}_{j},t).
\]
The computational efficiency is more prominent when $n$ is large, as the number of coefficients is significantly reduced from $n+m$ to $q+m$.


For the Sobolev space $\mathcal{W}_{2,m}$, the penalty function
$J(\cdot)$ is $J(f)=\int_{0}^{1}(f^{(m)}(s))^{2}ds$, and
($\ref{eq:penal lh}$) becomes 
\begin{eqnarray}
(\hat\theta , \hat{\beta}_{\lambda}) &=& \argmin_{\theta\in\mathbb{R}^p,\beta\in\mathcal{W}_{2,m}}-{\frac{1}{n}}\sum_{i=1}^{n}\Delta_{i}\big\{\eta_{\alpha}(W_{i})-\log\sum_{T_{j}>T_{i}}\exp(\eta_{\alpha}(W_{j})\big\} \nonumber\\
 & &\quad\quad\quad+\lambda\int_{0}^{1}(\beta^{(m)}(s))^{2}ds.\label{eq: Wm penal lh}
\end{eqnarray}
Let $\xi_{\nu}=t^{\nu-1}/(\nu-1)!,\,\nu=1,\ldots m,$ be the orthonormal
basis of  the null space 
\[
\mathcal{N}_{J}=\Big\{\beta\in\mathcal{W}_{2,m},\,\int_{0}^{1}(\beta^{(m)}(s))^{2}ds=0\Big\}.
\]
Write $G_{m}(t,u)=(t-u)_{+}^{m-1}/(m-1)!$, then the kernels are in
forms of 
\[
K_{0}(s,t)=\sum_{\nu=1}^{m}\xi_{\nu}(s)\xi_{\nu}(t), \ \text{\ and \ }\ K_{1}(s,t)=\int_{0}^{1}G_{m}(s,u)G_{m}(t,u)du.
\]
Hence, the estimator is given by 
\begin{equation}
\hat{\beta}_{\lambda}(t)=\sum_{\nu=1}^{m}d_{v}\xi_{\nu}(t)+\sum_{i=1}^{n}c_{i}\int_{0}^{1}X_{i}(s)K_{1}(s,t)ds.\label{eq:Wm est}
\end{equation}
We may obtain the constants $c_i$ and $d_j$ as well as the estimator $\hat{\theta}$ by maximizing the objective function (\ref{eq: Wm penal lh}) after plugging $\hat{\beta}_{\lambda}(t)$ back into the objective function. 

\subsection{Choosing the smoothing parameter}

The choice of the smoothing parameter $\lambda$ is always a critical but difficult question. In this section, we borrow ideas from \cite{Gu2013} and provide a simple  GCV method to choose $\lambda$.  The key idea is to draw an analogy between the partial likelihood
estimation and weighted density estimation, which then allows us to define a criterion
analogous to the Kullback-Leibler distance to select the best
performing  smoothing parameter.  Below we provide more details.


Let $i_{1},\ldots i_{N}$ be the index for the uncensored data, i.e
$\Delta_{i_{k}}=1$, for $k=1,\ldots N$ and $N=\sum_{1}^{n}\Delta_{i}$.   
Define weights $w_{i_k}(\cdot)$  as $w_{i_k}(t)=I\{t \geq T_{i_k}\}$ and \[
f_{\alpha|i_{k}}(t,w)=\frac{w_{i_{k}}(t)e^{\eta_{\alpha}(w)}}{\sum_{k=1}^N w_{i_{k}}(t)e^{\eta_{\alpha}(w)}}.
\]

Following the suggestion in Section 8.5 of \cite{Gu2013},  we  extend the Kullback-Leibler distance for density functions to  the partial likelihood as follows,


\begin{eqnarray*}
KL(\hat{\alpha}_{\lambda},\alpha)&=&  \frac{1}{N}\sum_{k=1}^{N}  \mathbb{E}_{f_{\alpha_{0}|i_k}} \Big\{\log\frac{f_{\alpha_{0}|i_k}(T_{i_k}, W_{i_k})}{f_{\hat{\alpha}|i_k} (T_{i_k}, W_{i_k})}\Big\}\\
&=&\frac{1}{N}\sum_{k=1}^{N}\mathbb{E}_{f_{\alpha_{0|i_{k}}}}\Big\{\log\frac{e^{\eta_{\alpha_{0}}(W_{i_{k}})}}{\sum_{j=1}^{n}w_{i_{k}}(T_{j})e^{\eta_{\alpha_{0}}(W_{j})}}-\log\frac{e^{\eta_{\hat{\alpha}_{\lambda}}(W_{i_{k}})}}{\sum_{j=1}^{n}w_{i_{k}}(T_{j})e^{\eta_{\hat{\alpha}_{\lambda}}(W_{j})}}\Big\}.
\end{eqnarray*}

Dropping off terms not involving $\hat{\alpha}_{\lambda}$, we have
a relative KL distance 
\[
RKL(\hat{\alpha}_{\lambda},\alpha)=-\frac{1}{N}\sum_{k=1}^{N}\mathbb{E}_{f_{\alpha_{0|i_{k}}}}\eta_{\hat{\alpha}_{\lambda}}(W)+\frac{1}{N}\sum_{k=1}^{N}\log\sum_{j=1}^{n}w_{i_{k}}(T_{j})e^{\eta_{\hat{\alpha}_{\lambda}}(W_{j})}.
\]
The second term is ready to be computed once we have an estimate $\hat{\alpha}_{\lambda}$,
but the first term involves $\alpha_{0}$ and needs to be estimated. We
approximate the RKL by 
\[
\widehat{RKL}(\hat{\alpha}_{\lambda},\alpha_{0})=-\frac{1}{n}\sum_{i=1}^{n}\eta_{\hat{\alpha}_{\lambda}}^{[i]}(W_{i})+\frac{1}{N}\sum_{i=1}^{n}\Delta_{i}\log\sum_{T_{j}\geq T_{i}}\exp\{\eta_{\hat{\alpha}_{\lambda}}(W_{j})\}.
\]

Based on this $\widehat{RKL}(\hat{\alpha}_{\lambda},\alpha_{0})$, a function 
$\mbox{GCV}(\lambda)$ can be derived analytically when replacing the penalized
partial likelihood function by its quadratic approximation,
\begin{eqnarray*}
GCV(\lambda)&=&-\frac{1}{n}\sum_{i=1}^{n}\eta_{\hat{\alpha}_{\lambda}}(W_{i})+\frac{1}{n(n-1)}\mbox{tr}[(SH^{-1}S)(\mbox{diag}\Delta-\Delta\boldsymbol{1}'/n)]\\
 & & \quad \quad \quad+\frac{1}{N}\sum_{i=1}^{n}\Delta_{i}\log\sum_{T_{j}\geq T_{i}}\exp\{\eta_{\hat{\alpha}_{\lambda}}(W_{j})\}.
\end{eqnarray*}
Details of deriving $\mbox{GCV}(\lambda)$ are given in Section $\ref{sub: GCV}$.



\subsection{Calculating the information bound $I(\theta)$}\label{Sec: ACE}

To calculate the information bound $I(\theta)$, we apply the ACE method \citep{Breiman1985}, the estimator
of which is shown to converge to $(a^{*},\, g^{*})$. For simplicity,
we take $Z$ as a one-dimensional scalar. When $Z$ is a vector, we just need to
apply the following procedure to all dimensions of $Z$ separately.

Theorem $\ref{thm:score function}$ shows that 
\[
I(\theta)=\mathbb{E}\{\Delta[Z-a^{*}(t)-\eta_{g^{*}}(X)]^{\otimes2}\}
\]
with $(a^{*},\, g^{*})\in{L}_{2}\times\mathcal{H}(K)$
being  the unique solution that minimizes
\[
\mathbb{E}\Big\{\Delta||Z-a(T)-\eta_{g}(X)||^{2}\Big\}.
\]
Furthermore, the proof of Theorem $\ref{thm:score function}$  reveals 
that this is equivalent to the following: $(a^{*},\, g^{*})$ is the unique
solution to the equations: 
\[
\mathbb{E}(Z-a^{*}-\eta_{g^{*}}|T,\Delta=1)=0,\ \ a.s.\, P_{T}^{(u)},
\]
\[
\mathbb{E}(Z-a^{*}-\eta_{g^{*}}|X,\Delta=1)=0,\ \ a.s.\, P_{X}^{(u)},
\]
where $P_{T}^{(u)}$ and $P_{X}^{(u)} $  represent, respectively, the measure space of $(T, \Delta=1) $ and $(X, \Delta=1)$.

The  idea of ACE is to update $a$ and $g$ alternatively until the
objective function $e(a,g)=\mathbb{E}\Delta||Z-a(T)-\eta_{g}(X)||^{2}$
stops to  decrease. In our case, the procedure is as follows: 
\begin{description}
\item [(i)] Initialize $a$ and $g$, 
\item [(ii)] Update $a$ by 
\[
a(T)=\mathbb{E}(Z-\eta_{g}|T,\Delta=1)=0,
\]
\item [(iii)] Update $g$ such that 
\[
\eta_{g}(X)=\mathbb{E}(Z-a|X,\Delta=1)=0,\ \ a.s.\, P_{X}^{(u)}, 
\]
\item [(iv)] Calculate $e(a,g)=\mathbb{E}\Delta||Z-a(T)-\eta_{g}(X)||^{2}$
and repeat (ii) and (iii) until $e(a,g)$ fails to decrease.
\end{description}

In practice, we replace $\mathbb{E}\Delta||Z-a(T)-\eta_{g}(X)||^{2}$
by the sample mean
\[
e(a,g)=\frac{1}{n}\sum_{i=1}^{n}\Delta_{i}||Z_{i}-a(T_{i})-\eta_{g}(X_{i})||^{2}.
\]
As for $a$ and $g$, we need to employ some smoothing techniques.    For a given  $g\in\mathcal{H}(K)$ we calculate 
\[
\tilde{a}_{i}=\sum_{T_{j}=T_{i}}\Delta_{j}[Z_{j}-\eta_{g}(X_{j})]/\sum_{T_{j}=T_{i}}\Delta_{j},
\]
and update $a(t)$ as the local polynomial regression estimator for the 
data $(T_{1},\tilde{a}_{1}),...,(T_{n},\tilde{a}_{n})$.  For a given  $a\in{L}_{2}$ we calculate 
\[
y_{i}=Z_{i}-a(T_{i}),\text{ for all }\Delta_{i}=1,
\]
and update $g$ by fitting a functional linear regression 
\[
y=\int g(s)X(s)ds+\epsilon,
\]
based on the data $(y_{i},X_{i})$ with $\Delta_{i}=1.$ More details
can be find in \cite{YuanCai2010}. When $(a^{*},\, g^{*})$ is obtained, $I(\theta)$
is estimated by 
\[
\widehat{I(\theta)}=\frac{1}{n}\sum_{i=1}^n\Delta_{i}[Z_{i}-a^{*}(T_{i})-\eta_{g^{*}}(X_{i})]^{\otimes2}.
\]

\section{Numerical Studies}
In this session, we first carry out simulations under different settings
to study the finite sample performance of the proposed method
and to demonstrate practical implications of the theoretical results. In the second part, we apply the proposed method to data that were collected to study the effect of early reproduction history to the longevity of female Mexican fruit flies.

\subsection{Simulations}

We adopt a similar design as that in \cite{YuanCai2010}. The functional covariate $X$ is generated
by a set of cosine basis functions, $\phi_{1}=1$ and $\phi_{k+1}(s)=\sqrt{2}\cos(k\pi s)$
for $k\geq1$, such that 
\[
X(s)=\sum_{k=1}^{50}\zeta_{k}U_{k}\phi_{k}(s),
\]
where the $U_{k}$ are independently sampled from the uniform distribution
on $[-3,3]$ and $\zeta_{k}=(-1)^{k+1}k^{-v/2}$ with $v=1,\,1.5,\,2,\,2.5$.
In this case, the covariance function of $X$ is $C(s,t)=\sum_{k=1}^{50}3k^{-v}\phi_{k}(s)\phi_{k}(t)$.
The coefficient function $\beta_{0}$ is 
\[
\beta_{0}=\sum_{i=1}^{50}(-1)^{k}k^{-3/2}\phi_{k},
\]
which is from a Sobolov space $\mathcal{W}_{2,2}$. The reproducing
kernel takes the  form:
\[
K(s,t)=1+st+\int_{0}^{1}(s-u)_{+}(t-u)_{+}du,
\]
and $K_{1}=\int_{0}^{1}(s-u)_{+}(t-u)_{+}du$. The null space becomes
$\mathcal{N}_{J}=\mbox{span}\{1,s\}$. The penalty function as mentioned
before is $J(f)=\int(f'')^{2}$. The vector covariate $Z$ is set to be
univariate  with distribution $\mathcal{N}(0,1)$ and corresponding
slope $\theta=1$. The failure time $T^{u}$ is generated based on
the hazard function 
\[
h(t)=h_{0}(t)\exp\Big\{\theta'Z+\int_{0}^{1}X(s)\beta_{0}(s)ds\Big\},
\]
where $h_{0}(t)$ is chosen as a constant or a linear function $t$.
Given $X$, $T^{u}$ follows an exponential distribution when $h_{0}$
is a constant, and follows a Weibull distribution when $h_{0}(t)=t$.
The censoring time $T^{c}$ is generated independently,  following an
exponential distribution with parameter $\gamma$ which controls the
censoring rate. When $h_{0}(t)$ is constant, $\gamma=19$ and $3.4$
lead to censoring rates around $10\%$ and $30\%$ respectively. Similar
censoring rates result from $\gamma=15$ and 3.9 for the case when
$h_{0}(t)=t$. 
$(T,\Delta)$ is then generated by $T=\min\{T^{u},\,T^{c}\}$ and $\Delta=I\{T^{u}\leq T^{c}\}$.

The criterion to evaluate the performance of the estimators $\hat{\beta}$
is the mean squared error, defined as 
\[
MSE(\hat{\beta})=\Big\{\frac{1}{\sum_{i=1}^{n}\Delta_{i}}\sum_{i=1}^{n}\Delta_{i}\Big(\eta_{\hat{\beta}}(X_{i})-\eta_{\beta_{0}}(X_{i})\Big)\Big\}^{1/2},
\]
which is an empirical version of $||\hat{\beta}-\beta_{0}||_{C_{\Delta}}$.
To study the trend as the sample size increases, we vary the sample size
$n$ according to $n=50,100,150,200$ for each value $v=1,\,1.5,\,2,\,2.5$.
For each combination  of censoring rate, $h_{0}$, $v$ and $n$, the simulation
is repeated $1000$ times, and the average mean squared error was obtained
for each scenario.

\begin{figure}
\centering
\includegraphics[width=5.7cm]{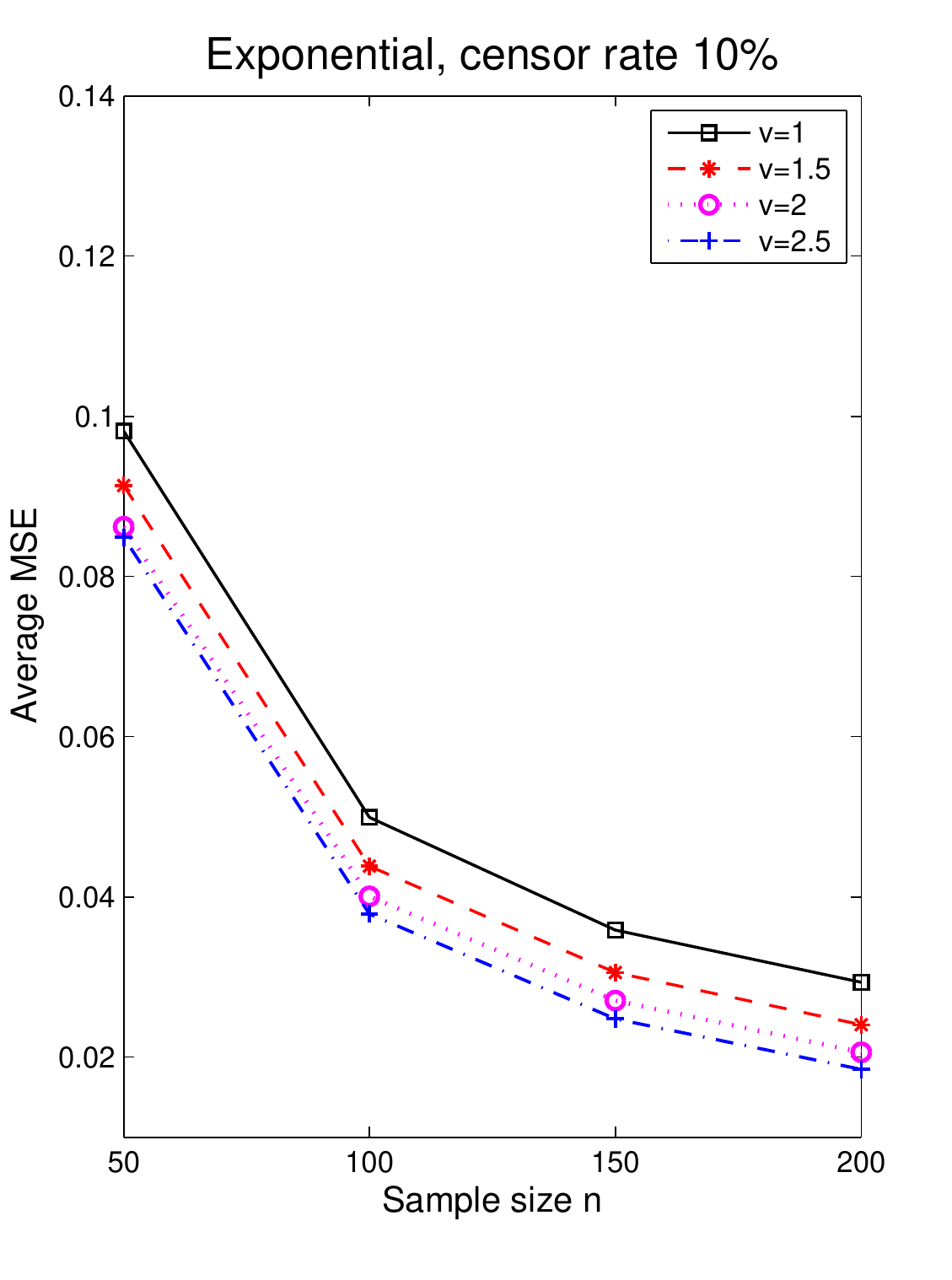}\includegraphics[width=5.7cm]{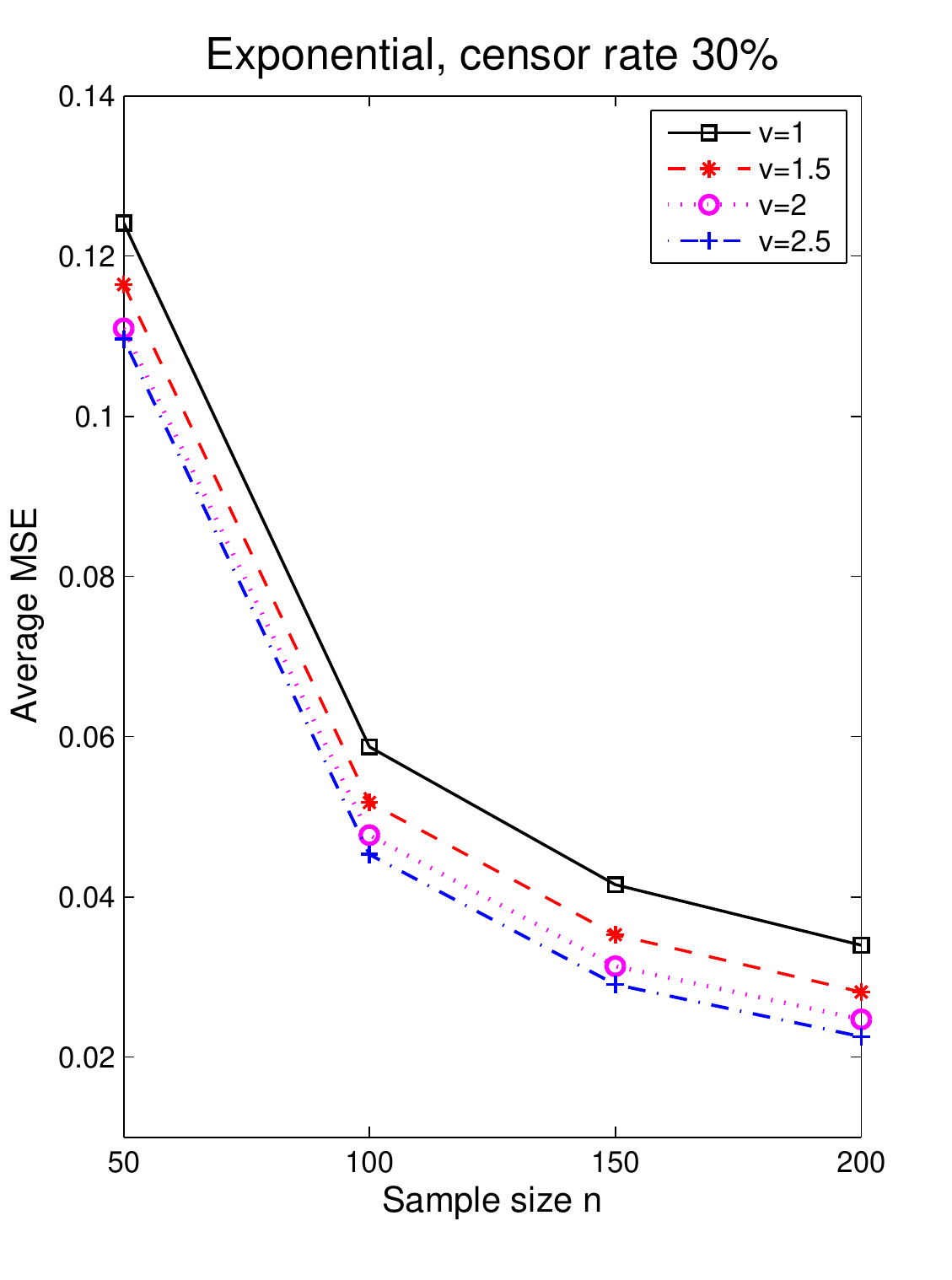}

\includegraphics[width=5.7cm]{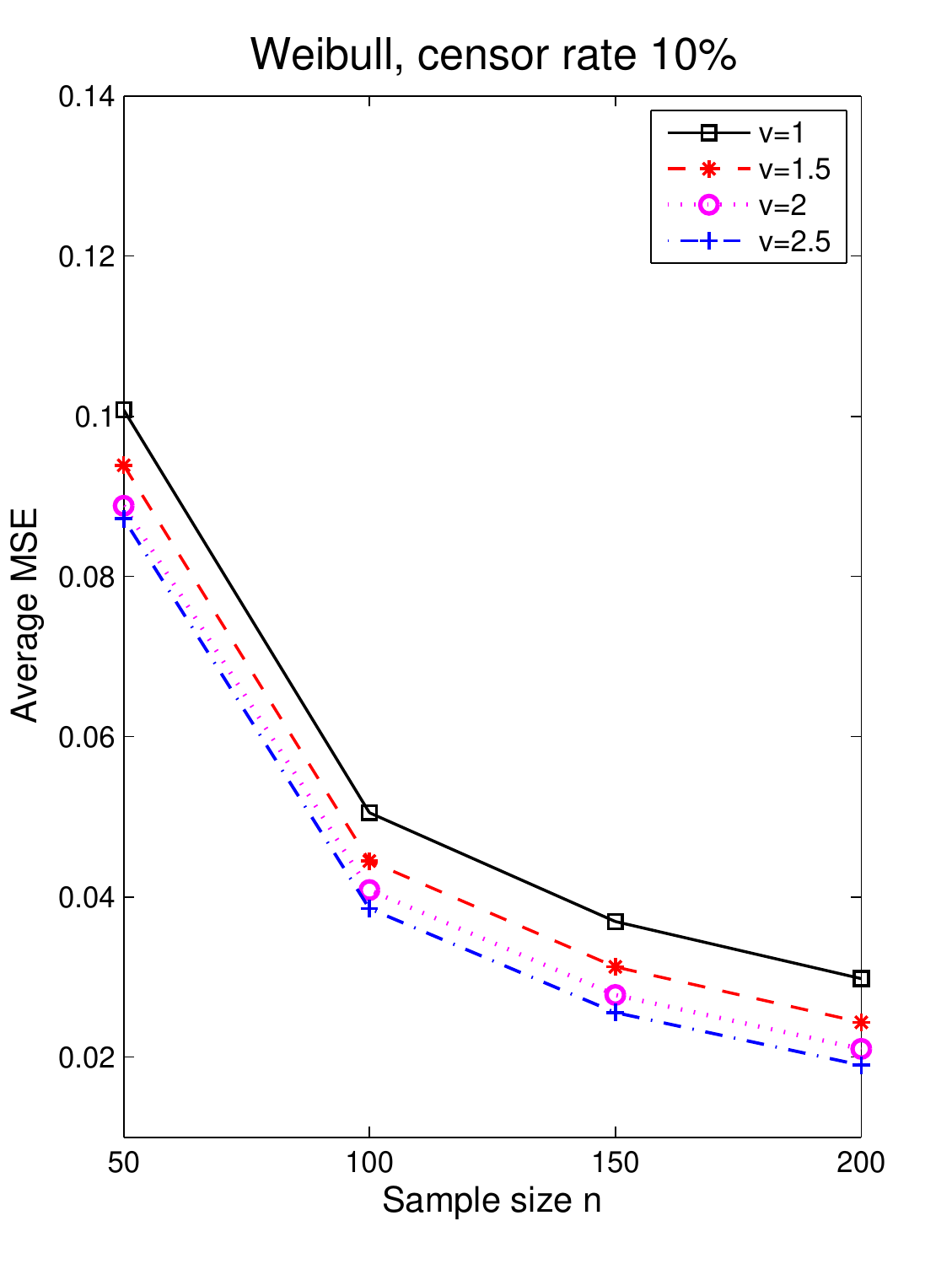}\includegraphics[width=5.7cm]{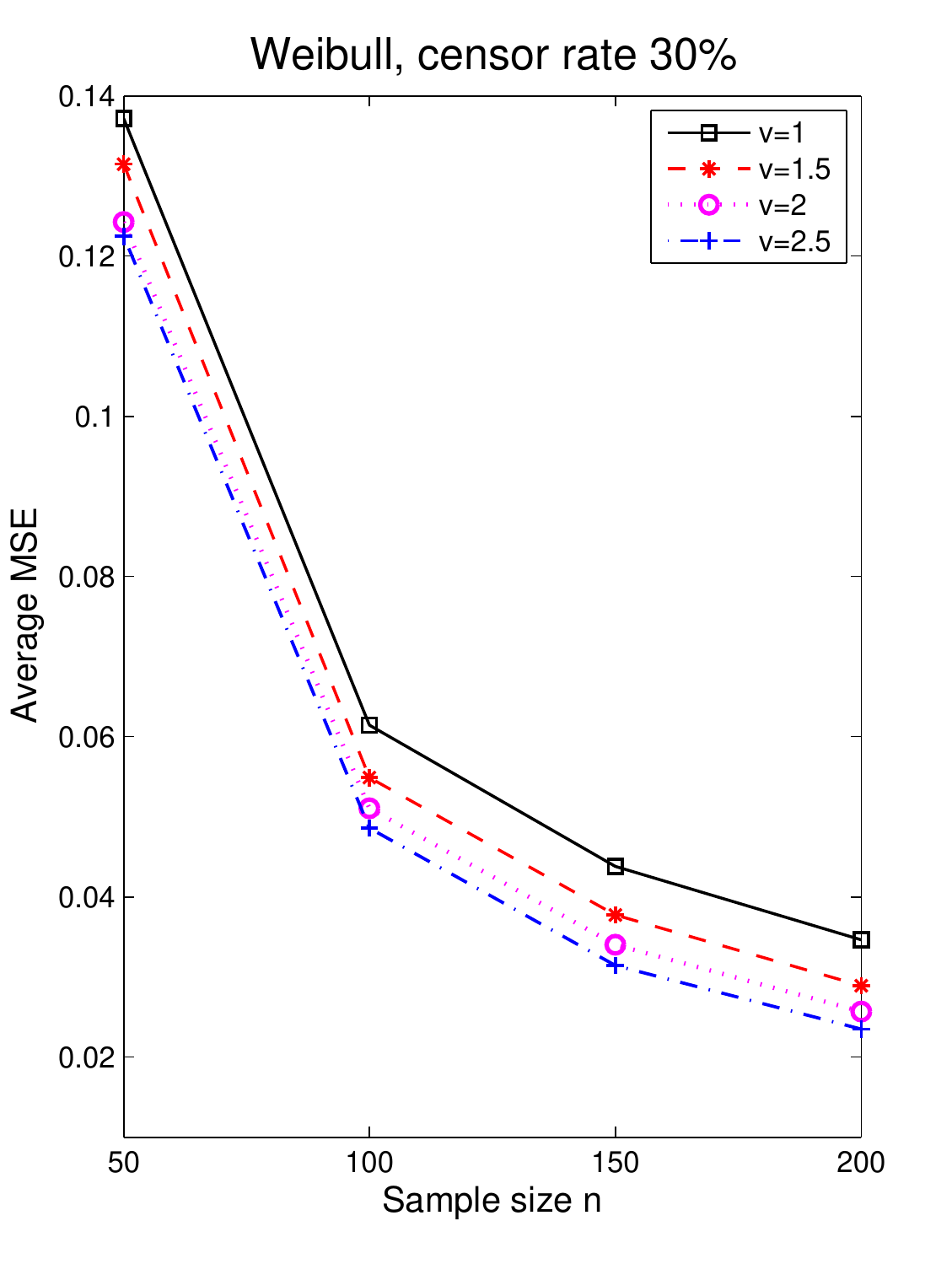}

\caption{\label{fig: sm set 1} The average MSE based on 1000 simulations.
The top panel is for the constant baseline hazard function and the
bottom panel is for the linear baseline hazard function. For each
panel, from left to right, the censoring rate is controlled to be
around 10\% and 30\%. The sample sizes are $n=50,100,150,200$ and
the decay rate parameters are $v=1,1.5,2,2.5$.}
\end{figure}

Note that for a fixed $\gamma$, $\mathbb{E}(\Delta|X)$ is roughly
a constant  for  different values of $v$. Therefore $C_{\Delta}(s,t)$
is approximately proportional to $C(s,t)=\sum_{k=1}^{50}k^{-v}\phi_{k}(s)\phi_{k}(t)$.
In this case, $v$ controls the decay rate of the eigenvalues of $C_{\Delta}$
and  $K^{1/2}C_{\Delta}K^{1/2}$.
It follows from Theorem \ref{thm: convergence rate} that a faster
decay rate of the eigenvalues leads to a faster convergence rate. Figure
\ref{fig: sm set 1} displays the average MSE based on 1000 simulations.
The simulation results are in agreement with Theorem \ref{thm: convergence rate};
it is very clear that when $v$ increases from 1 to 2.5 with the remaining
parameters fixed, the average MSEs decrease steadily. The average
MSEs also decrease with the sample sizes. Besides, for both the exponential
and Weibull distribution, the average MSEs are lower for each setting
at the 10\% censoring rate comparing to the values for the  30\% censoring rate.  This   is consistent with the expectation 
 that the lower the censoring rate is,
the more accurate the estimate will be.

Averages and standard deviations of the estimated $\hat{\theta}$, 
for each setting of $v$ and $n$ over 1000 repetition for the case
of $h_{0}=c$ and 30\% censoring rate,  are given in Table $\ref{tab:mtheta}$.
For each case of $v$, as $n$ increases, the average of $\hat{\theta}$
gets closer to the true value and the standard deviation
decreases. Noting that the results do not vary much across different
values of $v$, as $v$ is specially designed to examine the estimation
of $\beta$ and  has little effect on the estimation of  $\theta$. 

For each
simulated dataset, we also calculated the information bound $I(\theta)$ based on the ACE
method proposed in Section 3.3.  The inverse of this  information bound, as suggested by  Theorem $\ref{thm: asymp normal}$,  can be used to estimate the asymptotic variance of $\hat{\theta}$.  We further used these asymptotic variance estimates to 
construct a 95\% confidence
interval for $\theta$. 
Table $\ref{tab:vtheta}$
shows the observed percentage  the constructed 95\% confidence interval
covered the true value 1 for the various  settings. As expected, the covering
rates increase towards   95\% as $n$  gets larger. Results for other choices of  $h_{0}$ and censoring rates  were about the same and
are omitted.


\begin{table}
\caption{Average and standard deviation of $\hat{\theta}$.
($h_{0}=c$, 30\% censoring rate)}
\label{tab:mtheta}
\begin{centering}
\begin{tabular}{c>{\centering}b{2cm}>{\centering}b{2cm}>{\centering}b{2cm}>{\centering}b{2cm}}
\hline
n & $v=1$ & $v=1.5$ & $v=2$ & $v=2.5$\tabularnewline
\hline \hline
\multirow{2}{*}{50} & 1.061 & 1.064 & 1.064 & 1.065\tabularnewline
 & (0.264) & (0.265) & (0.264) & (0.265)\tabularnewline
\hline 
\multirow{2}{*}{100} & 1.027 & 1.030 & 1.031 & 1.031\tabularnewline
 & (0.164) & (0.164) & (0.164) & (0.163)\tabularnewline
\hline 
\multirow{2}{*}{150} & 1.013 & 1.016 & 1.017 & 1.018\tabularnewline
 & (0.133) & (0.132) & (0.131) & (0.131)\tabularnewline
\hline 
\multirow{2}{*}{200} & 1.011 & 1.013 & 1.015 & 1.016\tabularnewline
 & (0.111) & (0.111) & (0.110) & (0.110)\tabularnewline
\hline 
\end{tabular}
\par\end{centering}


\end{table}

\begin{table}
\begin{centering}
\caption{Covering rate of the 95\% confidence intervals for 
$\theta$. ($h_{0}=c$, 30\% censoring rate)}
\label{tab:vtheta}
\begin{tabular}{c>{\centering}p{2cm}>{\centering}p{2cm}>{\centering}p{2cm}>{\centering}p{2cm}}
\hline
n & $v=1$ & $v=1.5$ & $v=2$ & $v=2.5$\tabularnewline
\hline \hline
50 & 91.5\% & 91.9\% & 92.0\% & 91.5\%\tabularnewline
\hline 
100 & 93.3\% & 92.4\% & 92.4\% & 93.0\%\tabularnewline
\hline 
150 & 93.5\% & 93.1\% & 93.9\% & 93.4\%\tabularnewline
\hline 
200 & 93.6\% & 93.7\% & 93.9\% & 93.8\%\tabularnewline
\hline 
\end{tabular}
\par\end{centering}


\end{table}


\subsection{Mexican Fruit Fly Data}

We now apply the proposed method to the Mexican fruit fly data in 
 \cite{Carey2005}. There were 1152  female flies
in that paper coming from four cohorts, for illustration purpose we are using the data from
cohort 1 and cohort 2, which consist of the lifetime and daily reproduction
(in terms of number of eggs laid daily) of 576 female flies.

We are interested in whether and how early  reproduction  will affect the lifetime of female Mexican fruit flies.   
For this reason, we exclude  28 infertile flies from cohort 1 and     20 infertile flies from cohort 2. The period for early reproduction is chosen to be from day 6 to day 30 based on   the average reproduction curve (Figure
$\ref{fig: avenegg}$), which shows that no flies laid any eggs before day 6 and the peak of reproduction was day 30. 
Once the period of early reproduction was determined to be $[6, 30] $, we further excluded  flies that died before day 30 to guarantee a fully observed trajectory for all flies and this leaves us with a total of 479 flies for further exploration of the functional Cox model.  The mean and median lifetime of the remaining  224  flies in cohort 1 is   56.41 and 58 days respectively;  the mean and the median lifetime of the  remaining 255 flies in cohort 2 is  55.78  and  55 days respectively.  

The trajectories of early reproduction  for these 479 flies are of interest to researchers but they are very noisy, so  for visualization   we display the smoothed  egg-laying curves for the first 100 flies  (Figure $\ref{fig:Xt}$).   The  data of these 100 flies were individually smoothed with a local linear  smoother, but the subsequent data analysis for all 479 flies was based on the original data without smoothing. 

\begin{figure}
\centering{}\includegraphics[width=12cm]{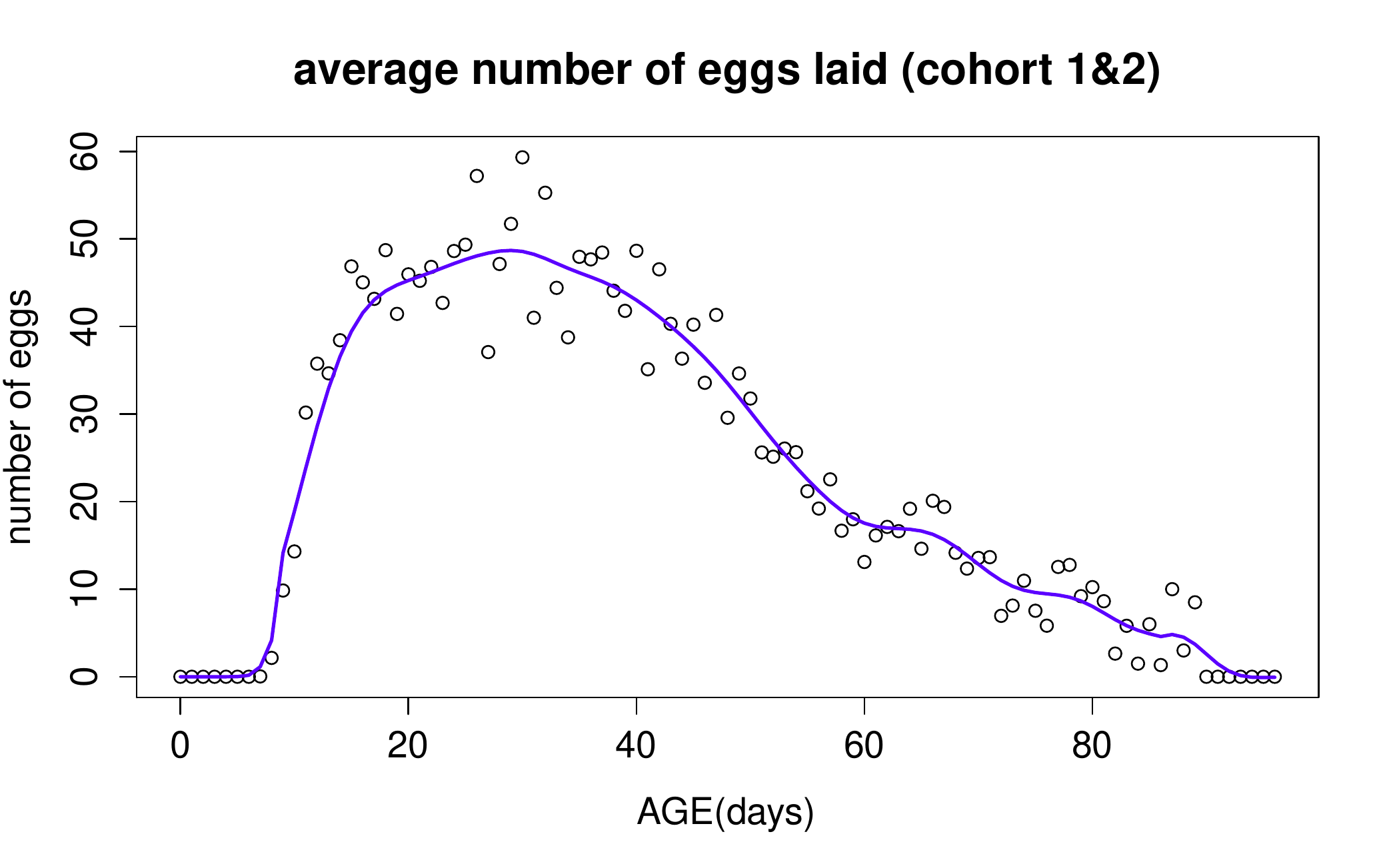}
\caption{\label{fig: avenegg} Average number of eggs laid daily for both cohorts}
\end{figure}


\begin{figure}
\begin{centering}
\includegraphics[width=10cm]{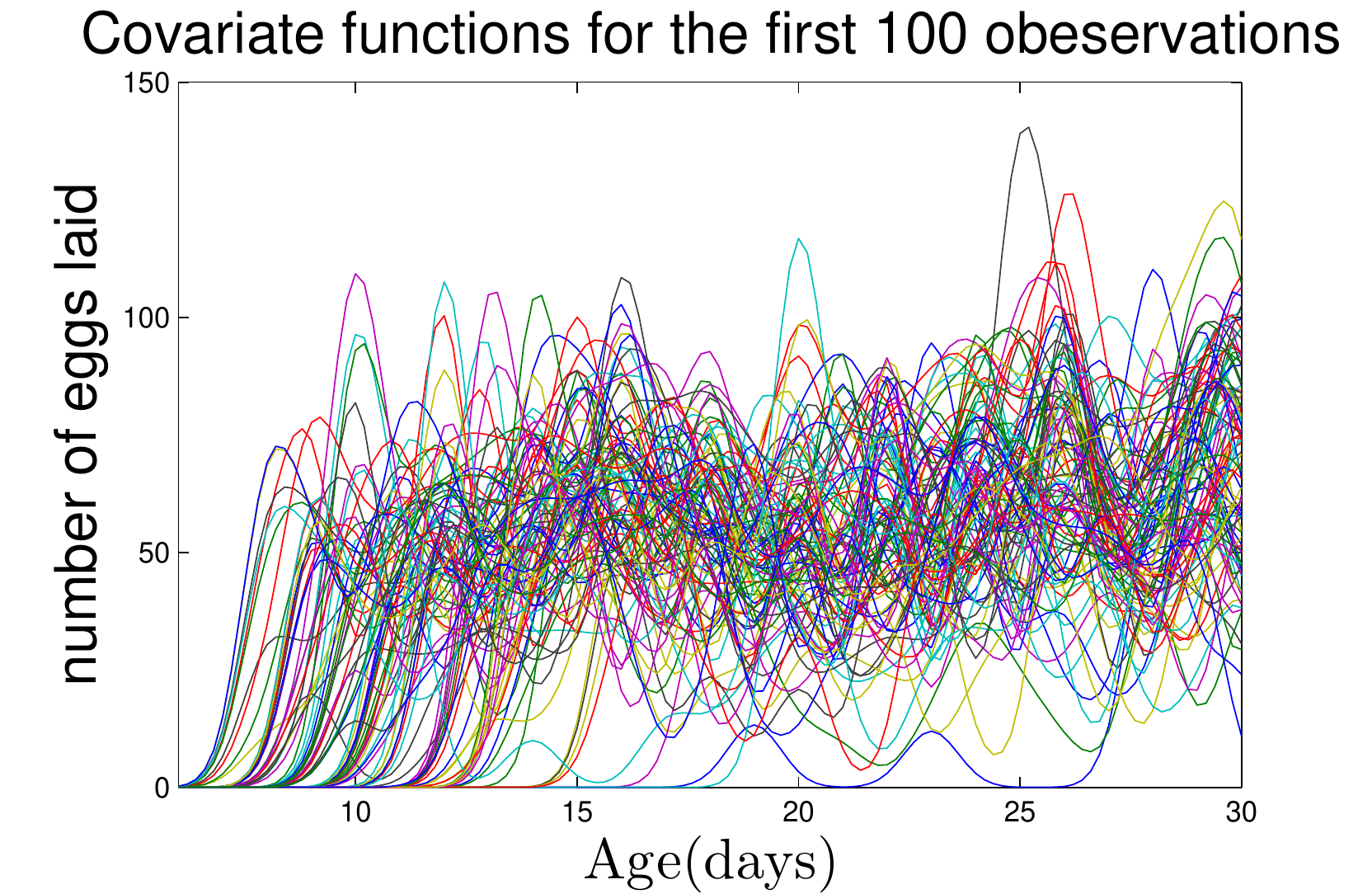}
\par\end{centering}

\caption{\label{fig:Xt}Pre-smoothed individual curves for the first 100 observations. }

\end{figure}

Using the original egg-laying curves from day 6 to day 30 as the 
longitudinal covariates and the cohort indicator as  a time-independent covariate,  the functional Cox model  resulted in an  estimate  $\hat{\theta}= 0.0562$ with 95\% confidence
interval $[-0.1235,0.2359]$. Since zero is included in the interval, we conclude that the cohort effect is not significant. Figure $\ref{fig: beta_nc}$
shows the estimated coefficient function $\hat{\beta}$ for the longitudinal covariate. The shaded area is the $95\%$ pointwise bootstrap confidence interval. Under the functional Cox model, a positive $\hat{\beta}(s)$
yields  a larger hazard function and a decreased probability of
survival and vice versa for a negative $\hat{\beta}(s)$. 

Checking
the plot of $\hat{\beta}(s)$, we can see that $\hat{\beta}(s)$ starts
with a large positive value, but  decreases fast to near zero
on day 13 and  stays around zero till day 22, then declines again mildly
towards day 30. The pattern of $\hat{\beta}(s)$ indicates that higher
early reproduction before day 13 results in a much higher mortality
rate suggesting the high cost of early reproduction,  whereas  a higher reproduction that occurs after day 22 tends to lead to
a relatively lower mortality rate, suggesting that  reproduction past day 22 might be sign of physical fitness.  However, the latter  effect is less significant than the early reproduction effect   as indicated by the bootstrap confidence interval. 
Reproduction  between day 13 and
day 22 does not have a major effect on the mortality rate. In other words,
flies that lay a lot of eggs in their early age (before day 13) and relatively fewer eggs 
 after day 22 tend
to die earlier, while those with the opposite pattern tend to have a longer
life span.

\begin{figure}
\begin{centering}
\includegraphics[width=10cm]{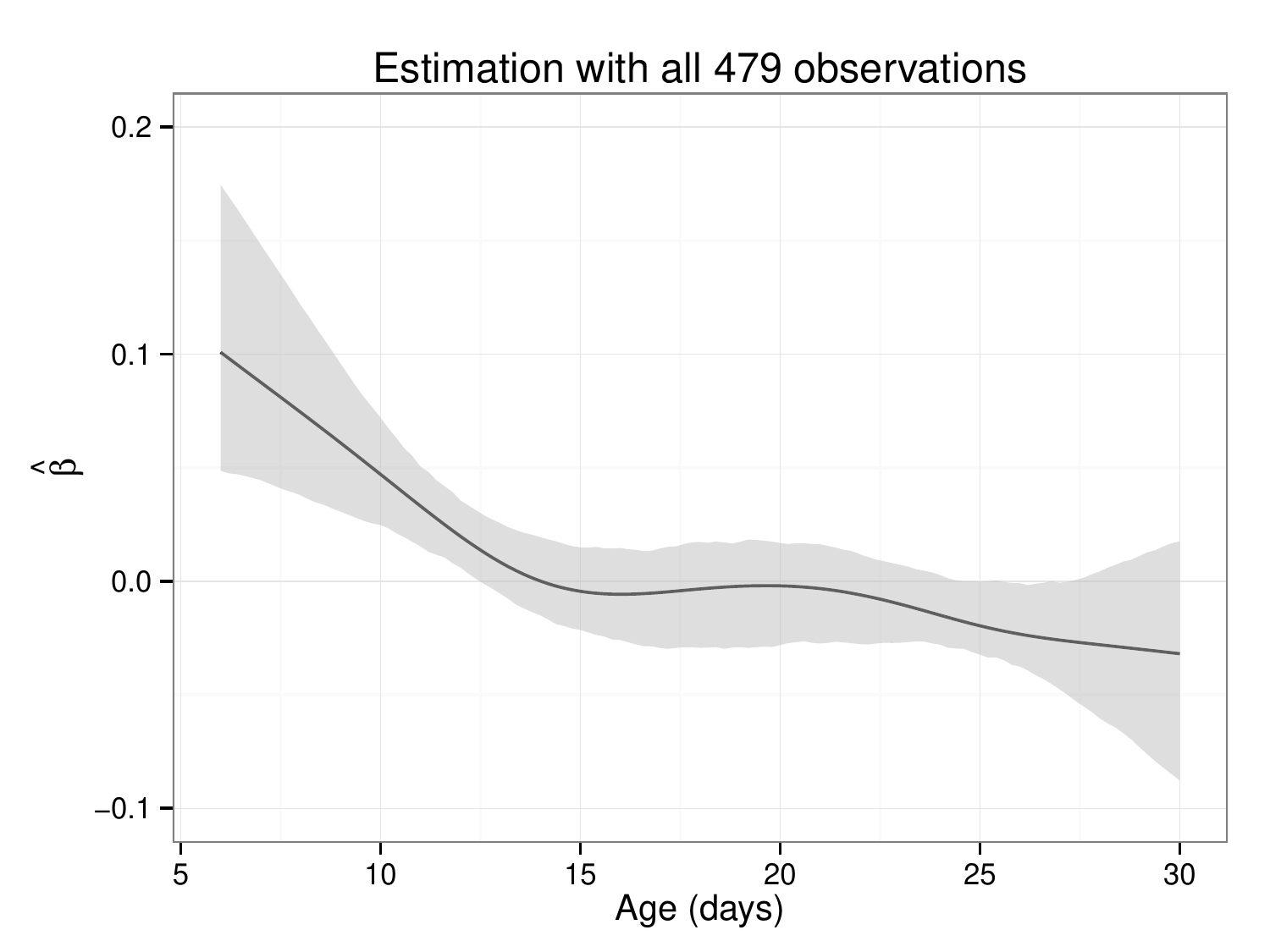}
\par\end{centering}

\caption{\label{fig: beta_nc}Estimated coefficient function $\hat{\beta}(s)$
using all 479 observations and 95\% pointwise c.i. for $\beta(s)$.}
\end{figure}

The Mexfly data contains no censoring,  so it is easy to check how the proposed method works in the presence of censored
data.  We artificially randomly  censor the data by 10\% and then again by 30\% using an exponential censoring distribution with parameter $\gamma=450$   and $150$, respectively.
The estimated coefficient $\hat{\theta}$
and corresponding 95\% confidence intervals are given in Table $\ref{tab: htheta and CI}$. 
Regardless of the censoring conditions, all the confidence intervals
contain  zero and therefore indicate a non-significant cohort effect. This
is consistent with the previous result for non-censored
data. The estimated coefficient functions $\hat{\beta}$ and the corresponding pointwise bootstrap confidence intervals are displayed 
in Figure $\ref{fig:betacs}$. Despite  the slightly different results
for different censoring proportions and choice of tuning parameters,
all the $\hat{\beta}$  have a similar  pattern. This indicates that
the proposed method is quite stable with respect to right censorship,
as long as the censoring rate is below 30\%.
%

\begin{table}

\caption{Values of fixed cut-off point and parameters for
generating random cut-off point, followed by the actual censored percentage
for both cohorts and the whole data.}
\label{tab: censor}
\begin{centering}
\begin{tabular}{|c|c|c|c|c|}
\hline 
 & \multicolumn{2}{c|}{fixed cut-off point} & \multicolumn{2}{c|}{random cut-off point}\tabularnewline
\hline 
 & $T^{c}=71$ & $T^{c}=62$ & $T^{c}\sim\exp(450)$ & $T^{c}\sim\exp(150)$\tabularnewline
 & (10\%) & (30\%) & (10\%) & (30\%)\tabularnewline
\hline 
Cohort 1 & 0.138 & 0.339 & 0.0.071 & 0.353\tabularnewline
\hline 
Cohort 2 & 0.067 & 0.259 & 0.110 & 0.251\tabularnewline
\hline 
Total & 0.100 & 0.296 & 0.092 & 0.300\tabularnewline
\hline 
\end{tabular}
\par\end{centering}

\end{table}

\begin{table}

\caption{The estimated $\hat{\theta}$ and  95\%
confidence interval for $\theta$ under different censoring conditions.}
\label{tab: htheta and CI}
\begin{centering}
\begin{tabular}{|c|c|c|}
\hline 
 & 10\% censoring & 30\% censoring\tabularnewline
\hline 
\multirow{2}{*}{fixed cut-off point} & 0.0929 & 0.0757\tabularnewline
 & {[}-0.0914, 0.2772 {]} & {[}-0.1268 0.2870{]}\tabularnewline
\hline 
\multirow{2}{*}{random cut-off point} & 0.0104 & 0.1863\tabularnewline
 & {[}-0.1705, 0.1913{]} & {[}-0.0177,0.3903{]}\tabularnewline
\hline 
\end{tabular}
\par\end{centering}

\end{table}

\begin{figure}
\begin{center}
\includegraphics[width=13cm]{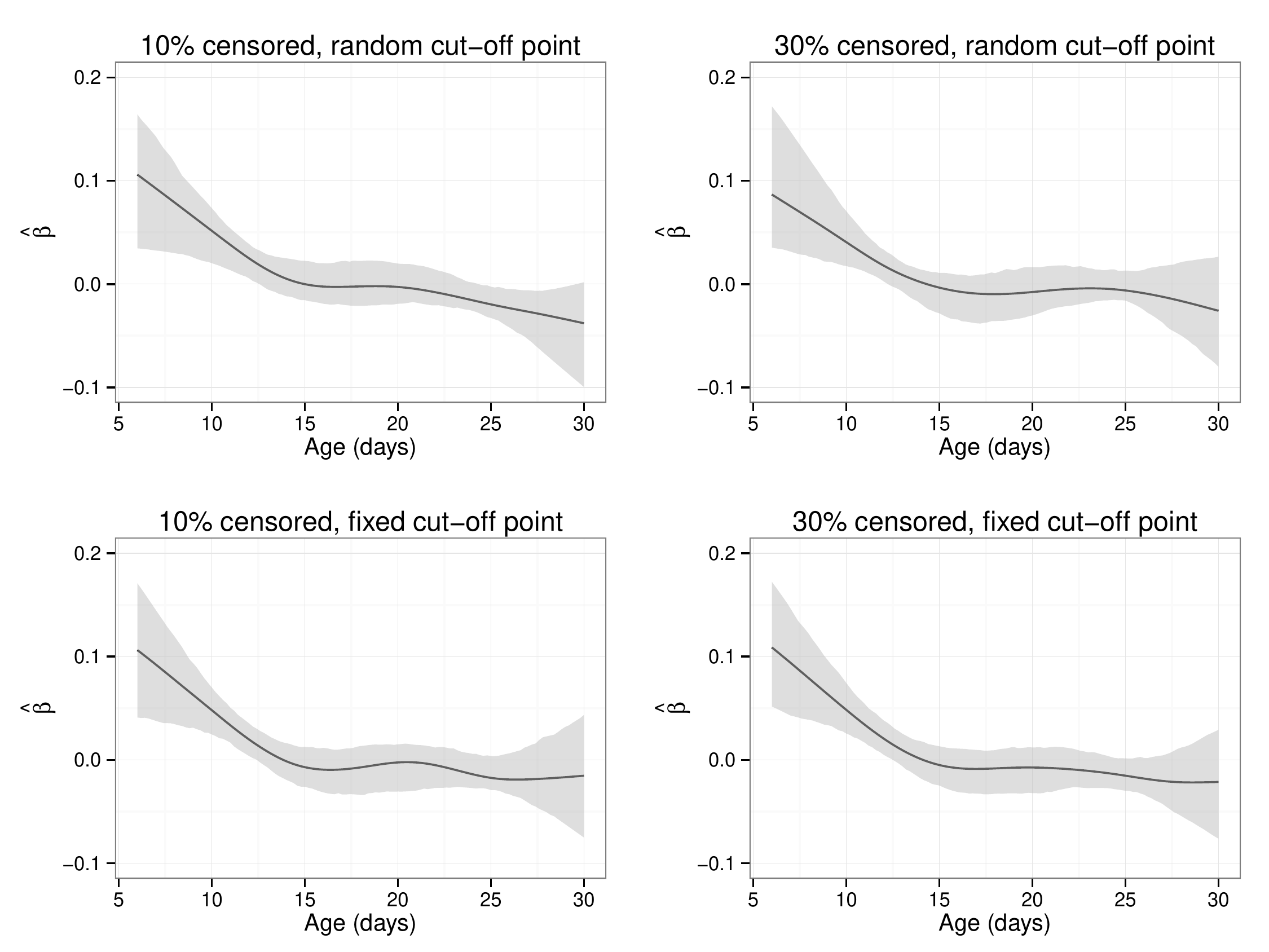}
\end{center}
\caption{\label{fig:betacs}Estimation for $\beta (s)$ with censored data and 95\% pointwise c.i.}
\end{figure}



\section{Proofs of Theorems}
We first introduce some notations by denoting $d(\beta_{1},\beta_{2})=||\beta_{1}-\beta_{2}||_{C_{\Delta}}$,
for any $\beta_{1},\beta_{2}\in\mathcal{H}(K)$;  $Y(t)=1_{\{T\geq t\}}$;
$Y_{j}(t)=1_{\{T_{j}\geq t\}}$, $1\leq j\leq n$;
and $\eta_{\beta}(X_{i})=\int_0^1 \beta(s)X_{i}(s)ds$.

Recall that  $W=(Z,\, X)$  represents the covariates,  $\alpha=(\theta,\beta)$ represents the corresponding regression coefficient with $\theta$   the coefficient for $Z$ and $\beta$  the coefficient function for $X(\cdot)$, and  the true   coefficient is denoted as $\alpha_{0}=(\theta_{0}, \beta_{0}).$
 The index $\eta_{\alpha}(W)=\theta'Z+\int_0^1 \beta(s)X(s)ds$ summarizes the information carried by the covariate $W$. 
To measure the distance between two coefficients $\alpha_1$ and $\alpha_2$ we use 
\[
d(\alpha_{1},\alpha_{2})^{2}=\mathbb{E} (\Delta[\eta_{\alpha_{1}}(W)-\eta_{\alpha_{2}}(W)]^{2}). 
\]

 Furthermore, we denote
\[
S_{0n}(t,\alpha)=\frac{1}{n}\sum_{j=1}^{n}Y_{j}(t)e^{\eta_{\alpha}(W_{j})},\ \ \ \ S_{0}(t,\alpha)=\mathbb{E}\{Y(t)e^{\eta_{\alpha}(W)}\}, 
\]
and for $\widetilde{\alpha}\in{L}_{2}\times\mathcal{H}(K)$,
\[
S_{1n}(t,\alpha)[\widetilde{\alpha}]=\frac{1}{n}\sum_{j=1}^{n}Y_{j}(t)e^{\eta_{\alpha}(W_{j})}\eta_{\widetilde{\alpha}}(W_{j})),\ \ \ \ S_{1}(t,\alpha)[\widetilde{\alpha}]=\mathbb{E}[Y(t)e^{\eta_{\alpha}(W)}\eta_{\widetilde{\alpha}}(W)].
\]
Define
\[
m_{n}(t,W,\alpha)=[\eta_{\alpha}(W)-\log S_{0n}(t,\alpha)]1_{\{0\leq t\leq\tau\}}, 
\]
and
\[
m_{0}(t,W,\alpha)=[\eta_{\alpha}(W)-\log S_{0}(t,\alpha)]1_{\{0\leq t\leq\tau\}}.
\]

Let $P_{n}$ and $P$ be the
empirical and probability measure of $(T_{i},\Delta_{i},W_{i})$ and
$(T,\Delta,W)$, respectively, and $P_{\Delta n}$ and $P_{\Delta}$
be the subprobability measure with $\Delta_{i}=1$ and $\Delta=1$
accordingly. The logarithm of the partial likelihood is $M_{n}(\alpha)=P_{\Delta n}m_{n}(\cdot,\alpha)$.  Let
$M_{0}(\alpha)=P_{\Delta}m_{0}(\cdot,\alpha)$.
Note that $P_{\Delta}$ is restricted to $T\in[0,\tau]$ due to the $1\{0\leq t\leq\tau\}$ term.

A useful identity due to Lemma 2 in \cite{Sasieni1992b} is 
\begin{equation}
\frac{S_{1}(t,\alpha)[\widetilde{\alpha}]}{S_{0}(t,\alpha)}=\mathbb{E}[\eta_{\widetilde{\alpha}}(W)|\, T=t,\Delta=1].\label{eq:ssieni}
\end{equation}

\subsection{Proof of Theorem $\ref{thm:score function}$}

The log-likelihood for a single sample $(t,\Delta,Z,X(\cdot))$ is 
\[
l(h_{0},\theta,\beta)=\Delta[\log h_{0}(t)+Z'\theta+\int_0^1 X(s) \beta(s) ds]-H_{0}(t)\exp[Z'\theta+\int_0^1 X(s) \beta(s) ds],
\]
where $H_{0}(t)=\int_{0}^{t}h_{0}(u)du$ is the baseline cumulative
hazard function. Consider a parametric and smooth sub-model $\{h_{(\mu_{1})}:\mu_{1}\in\mathbb{R}\}$
satisfying $h_{(0)}=h_{0}$ and 
\[
\frac{\partial\log h_{(\mu_{1})}}{\partial\mu_{1}}(t)\Big|_{\mu_{1}=0}=a(t).
\]
Let $\eta_{(\mu_{2})}(X)=\eta_{\beta}(X)+\eta_{\mu_{2}g}(X),$ for
$g\in\mathcal{H}(K)$. Therefore $\eta_{(0)}=\eta_{\beta}(X)$ and
\[
\frac{\partial\eta(\mu_{2})}{\partial\mu_{2}}(X)\Big|_{\mu_{2}=0}=\eta_{g}(X).
\]
Recall that $r(W)=\exp(\eta_{\alpha}(W))$, and $M(t)$ is the counting process martingale
associated with model (1),
\[
M(t)=M(t|W)=\Delta I\{T\leq t\}-\int_{0}^{t}I\{T\geq u\}r(W)dH_{0}(u).
\]
The score operators for the cumulative hazard $H_{0}$, coefficient function $\beta$, 
and the score vector for $\theta$ are the partial derivatives of
the likelihood $l(h_{(\mu_{1})},\theta,\eta_{(\mu_{2})})$ with respect
to $\mu_{1}$, $\mu_{2}$ and $\theta$ evaluated at $\mu_{1}=\mu_{2}=0$,
\[
i_{H}a:=\Delta a(T)-r(W)\int_{0}^{\infty}Y(t)a(t)dH_{0}(t)=\int_{0}^{\infty}a(t)dM(t),
\]
\[
i_{\beta}g:=\eta_{g}(X)[\Delta-r(W)H_{0}(T)]=\int_{0}^{\infty}\eta_{g}(X)\, dM(t),
\]
\[
i_{\theta}:=Z[\Delta-r(W)H_{0}(T)]=\int_{0}^{\infty}Z\, dM(t).
\]
Define $L(P_{T}^{(u)}):=\{a\in {\cal L}_2:\, \mathbb E[\Delta a^{2}(T)]<\infty\}$ and $L(P_X^{(u)}):=\{g\in {\cal H}(K): \mathbb E[\Delta \eta_g(X)]=0; \mathbb E[\Delta \eta_g^2(X) < \infty]\}$. Let
\[
A_{H}=\{i_{H}a:\, a\in L(P_{T}^{(u)})\},
\]
and 
\[
G=\{i_{\beta}g:\, g\in L(P_X^{(u)})\}.
\]
To calculate the information bound for $\theta$, we need to find
the (least favorable) direction $(a^{*},g^{*})$ such that $i_{\theta}-i_{H}a^{*}-i_{\beta}g^{*}$
is orthogonal to the sum space $\boldsymbol{A}=A_{H}+G$. That is,
$(a^{*},g^{*})$ must satisfy 
\[
\mathbb{E}[(i_{\theta}-i_{H}a^{*}-i_{\beta}g^{*})i_{H}a]=0,\ \ a\in L(P_{T}^{(u)}),
\]
\[
\mathbb{E}[(i_{\theta}-i_{H}a^{*}-i_{\beta}g^{*})i_{\beta}g]=0,\ \ g\in L(P_X^{(u)}).
\]
Following the  proof of Theorem 3.1 in \cite{Huang1999}, we can show that
$(a^{*},g^{*})$ satisfies 
\begin{equation}
\mathbb{E}[\Delta(Z-a^{*}-\eta_{g^{*}})a]=0,\ \ a\in L(P_{T}^{(u)}),\label{eq: a*}
\end{equation}
\begin{equation}
\mathbb{E}[\Delta(Z-a^{*}-\eta_{g^{*}})\eta_{g}]=0,\ \ g\in L(P_X^{(u)}).\label{eq:g*}
\end{equation}
Therefore, $(a^*, g^*)$ is the solution to the following equations:
\[
\mathbb{E}(Z-a^{*}-\eta_{g^{*}}|T,\Delta=1)=0,\ \ a.s.\, P_{T}^{(u)},
\]
\[
\mathbb{E}(Z-a^{*}-\eta_{g^{*}}|X,\Delta=1)=0,\ \ a.s.\, P_{X}^{(u)}.
\]
So, $(a^*, g^*)\in L(P_{T}^{(u)}) \times L(P_X^{(u)})$ minimizes 
\begin{equation}\label{equ:ace}
\mathbb E\Big\{ \Delta \big\| Z - a(T) - \eta_g(X) \big\|^2 \Big\}.
\end{equation}
It follows from Conditions A3 and A4 that the space $L(P_{T}^{(u)}) \times L(P_X^{(u)})$ is closed, so that the minimizer of (\ref{equ:ace}) is well-defined. Further, the solution can be obtained by the population version of the ACE algorithm of \cite{Breiman1985}. 


\subsection{Proof of Theorem $\ref{thm: convergence rate}$}

For some large number $M$, such that $||\theta_{0}||_{\infty}<M$
and $||\beta_{0}||_{K}<M$, define $\mathbb{R}_{M}=\{\theta\in\mathbb{R}^{p},\,||\theta||_{\infty}<M\}$
and $\mathcal{H}^{M}=\{\beta\in\mathcal{H}(K),||\beta||_{K}<M\}$.
Let $\alpha^{M}=(\theta^{M},\beta^{M})$ be the penalized partial likelihood
estimator with minimum taken over ${L}^{M}\times\mathcal{H}^{M}$,
i.e.
\begin{equation}
\alpha^{M}=\arg\min_{\alpha\in\mathbb{R}_{M}\times\mathcal{H}^{M}}-n^{-1}\sum_{i=1}^{n}\Delta_{i}\big\{\eta_{\alpha}(W_{i})-\log\sum_{T_{j}>T_{i}}\exp\{\eta_{\alpha}(W_{j})\}\big\}+\lambda\cdot J(\beta)\label{eq:beta_M}.
\end{equation}

We first prove that 
\begin{equation}
\sup_{\alpha\in\mathbb{R}_{M}\times\mathcal{H}^{M}}|M_{n}(\alpha)-M_{0}(\alpha)|\stackrel{P}{\rightarrow}0.  \label{eq:consist 1}
\end{equation}
Observe that 
\begin{eqnarray*}
 & &|M_{n}(\alpha)-M_{0}(\alpha)| \\
 & & \leq  |P_{\Delta n}m_{n}(\cdot,\alpha)-P_{\Delta n}m_{0}(\cdot,\alpha)|+|P_{\Delta n}m_{0}(\cdot,\alpha)-P_{\Delta}m_{0}(\cdot,\alpha)|\\
 & &\leq  P_{\Delta n}|\log S_{0n}(T,\alpha)-\log S_{0}(T,\alpha)|1_{\{0\leq T\leq\tau\}}+|(P_{n}-P)\Delta m_{0}(\cdot,\alpha)|\\
 & &\lesssim \sup_{0\leq t\leq\tau}|S_{0n}(t,\alpha)-S_{0}(t,\alpha)|+|(P_{n}-P)\Delta m_{0}(\cdot,\alpha)|\\
 &  &= \sup_{0\leq t\leq\tau}|(P_{n}-P)Y(t)e^{\eta_{\alpha}(W)}|+|(P_{n}-P)\Delta m_{0}(\cdot,\alpha)|.
\end{eqnarray*}
Lemma $\ref{lem:consist 1}$ shows that $\mathcal{F}_{1}=\{\Delta m_{0}(t,W,\alpha):\alpha\in\mathbb{R}_{M}\times\mathcal{H}^{M}\}$
and $\mathcal{F}_{2}=\{Y(t)e^{\eta_{\alpha}(W)}:\alpha\in\mathbb{R}_{M}\times\mathcal{H}^{M},\ 0\leq t\leq\tau\}$
are P-Glivenko-Cantelli, which means that both terms on the righthand side  above  converge to zero in probability uniformly with respect to $\alpha\in\mathbb{R}_{M}\times\mathcal{H}^{M}$.
Therefore ($\ref{eq:consist 1}$) holds.

The definition of $\alpha^{M}$ in ($\ref{eq:beta_M}$) indicates that
\[
-M_{n}(\alpha^{M})+\lambda J(\beta^{M})\leq-M_{n}(\alpha_{0})+\lambda J(\beta_{0}).
\]
Rearranging  the inequality with $M_{n}(\alpha^{M})$ on one side and the
fact that $\lambda\rightarrow0$ as $n\rightarrow\infty$ lead to 
\begin{equation}
M_{n}(\alpha^{M})\geq M_{n}(\alpha_{0})-o_{p}(1).\label{eq:consist 2}
\end{equation}

On the other hand, lemma $\ref{lem bound 1}$ implies  that $\sup_{d(\alpha,\alpha_{0})\geq\epsilon}M_0(\alpha)<M_0(\alpha_{0})$.
Combining this with ($\ref{eq:consist 1}$) and ($\ref{eq:consist 2}$) and by the consistency result in \citet[Theorem 5.7
on Page 45]{Vaart2000}, we can show that $\alpha^{M}$ is consistent, i.e. $d(\alpha^{M},\alpha_{0})\stackrel{P}{\rightarrow}0$.

Part (i) now follows from 
\[
d(\hat{\alpha},\alpha_{0})\leq d(\hat{\alpha}, \alpha^{M})+d(\alpha^{M},\alpha_{0}), 
\]
and $P(\hat{\alpha}=\alpha^{M})=P(||\hat{\beta}||_{K}<M,||\hat{\theta}||_{\infty}<M)\rightarrow1$,
as $M\rightarrow\infty$, i.e. $d(\hat{\alpha},\alpha^{M})\rightarrow0\, a.s.$.

For part (ii), we follow the proof of Theorem 3.4.1 in \cite{Vaart1996}. We first show that
\begin{equation}
E^{*}\sup_{\delta/2\leq d(\alpha,\alpha_{0}))\leq\delta}\sqrt{n}|(M_{n}-M_{0})(\alpha-\alpha_{0})|\lesssim\phi_{n}(\delta), \label{eq: conv rate cond 1}
\end{equation}
where $\phi_{n}(\delta)=\delta^{\frac{2r-1}{2r}}$. Direct calculation
yields that

\begin{eqnarray*}
& &(M_{n}-M_{0})(\alpha-\alpha_{0}) \\
&  &= P_{\Delta n}m_{n}(\cdot,\alpha)-P_{\Delta n}m_{n}(\cdot,\alpha_{0})-P_{\Delta}m_{0}(\cdot,\alpha)+P_{\Delta}m_{0}(\cdot,\alpha_{0})\\
 & &=  (P_{\Delta n}-P_{\Delta})(m_{0}(\cdot,\alpha)-m_{0}(\cdot,\alpha_{0}))\\
 &  & \ \ \ \ +P_{\Delta n}(m_{n}(\cdot,\alpha)-m_{n}(\cdot,\alpha_{0})-m_{0}(\cdot,\alpha)+m_{0}(\cdot,\alpha_{0}))\\
 & & =  (P_{\Delta n}-P_{\Delta})(m_{0}(\cdot,\alpha)-m_{0}(\cdot,\alpha_{0}))\\
 &  & \ \ \ \ \ \ \ \ +P_{\Delta n}(\log\frac{S_{0}(T,\alpha)}{S_{0}(T,\alpha_{0})}-\log\frac{S_{0n}(T,\alpha)}{S_{0n}(T,\alpha_{0})})\\
 &  &= I+II.
\end{eqnarray*}

For the first term, $I=(P_{\Delta n}-P_{\Delta})(m_{0}(\cdot,\alpha)-m_{0}(\cdot,\alpha_{0}))$.
By Lemma $\ref{lem: conv I and III}$, we have 
\[
\sup_{\delta/2\leq d(\alpha,\alpha_{0}))\leq\delta}|I|=O(\delta^{\frac{2r-1}{2r}}n^{-1/2}).
\]

For the second term $II$, we have
\begin{eqnarray*}
& &\sup_{\delta/2\leq d(\alpha,\alpha_{0}))\leq\delta}|II| \\
& \leq & \sup_{\begin{array}{c}
\delta/2\leq d(\alpha,\alpha_{0}))\leq\delta\\
t\in[0,\tau]
\end{array}}\Big|\log\frac{S_{0}(t,\alpha)}{S_{0}(t,\alpha_{0})}-\log\frac{S_{0n}(t,\alpha)}{S_{0n}(t,\alpha_{0})}\Big|\\
 & \leq & \sup_{\begin{array}{c}
\delta/2\leq d(\alpha,\alpha_{0}))\leq\delta\\
t\in[0,\tau]
\end{array}}c\Big|\frac{S_{0n}(t,\alpha)}{S_{0n}(t,\alpha_{0})}-\frac{S_{0}(t,\alpha)}{S_{0}(t,\alpha_{0})}\Big|\\
 & = & \sup_{\begin{array}{c}
\delta/2\leq d(\alpha,\alpha_{0}))\leq\delta\\
t\in[0,\tau]
\end{array}}c\Big|\frac{S_{0n}(t,\alpha)S_{0}(t,\alpha_{0})-S_{0n}(t,\alpha_{0})S_{0}(t,\alpha)}{S_{0}(t,\alpha_{0})S_{0n}(t,\alpha_{0})}\Big|.
\end{eqnarray*}
For $t\in[0,\tau]$, the denominator $S_{0}(t,\alpha_{0})S_{0n}(t,\alpha_{0})$
is bounded away from zero with probability tending to one. 
The numerator satisifes 
\begin{eqnarray*}
& &S_{0n}(t,\alpha)S_{0}(t,\alpha_{0})-S_{0n}(t,\alpha_{0})S_{0}(t,\alpha) \\
& &=S_{0}(t,\alpha_{0})[S_{0n}(t,\alpha)-S_{0n}(t,\alpha_{0})-S_{0}(t,\alpha)+S_{0}(t,\alpha_{0})]\\
 & &\quad \quad-[S_{0n}(t,\alpha_{0})-S_{0}(t,\alpha_{0})][S_{0}(t,\alpha)-S_{0}(t,\alpha_{0})].
\end{eqnarray*}
For the first term on the right side, we have  $S_{0}(t,\alpha_{0})=O(1)$
and 
\begin{eqnarray*}
& &[S_{0n}(t,\alpha)-S_{0n}(t,\alpha_{0})-S_{0}(t,\alpha)+S_{0}(t,\alpha_{0})]\\
& & \quad\quad\quad=(P_{n}-P)\big\{ Y(t)\big[\exp(\eta_{\alpha}(W))-\exp(\eta_{\alpha_{0}}(W))\big]\big\}.
\end{eqnarray*}

Define the above $(P_{n}-P)\big\{ Y(t)\big[\exp(\eta_{\alpha}(W))-\exp(\eta_{\alpha_{0}}(W))\big]\big\}\stackrel{def}{=}III$.

Lemma $\ref{lem: conv I and III}$ implies  that 
\[
\sup_{\delta/2\leq d(\alpha,\alpha_{0}))\leq\delta}|III|=O(\delta^{\frac{2r-1}{2r}}n^{-1/2}).
\]
For the second term, the Central Limit Theorem implies $S_{0n}(t,\alpha_{0})-S_{0}(t,\alpha_{0})=O_{p}(n^{-1/2})$, 
and

\begin{eqnarray*}
|S_{0}(t,\alpha)-S_{0}(t,\alpha_{0})| & \leq & E\big\{ Y(t)\big|\exp(\eta_{\alpha}(W))-\exp(\eta_{\alpha_{0}}(W))\big|\big\}\\
 & \lesssim & \Big(E[\eta_{\alpha}(W)-\eta_{\alpha_{0}}(W)]^{2}\Big)^{1/2}\\
 & \lesssim & d(\alpha,\alpha_{0}).
\end{eqnarray*}
Therefore 
\[
\sup_{\delta/2\leq d(\alpha,\alpha_{0}))\leq\delta}|II|\leq O(\delta^{\frac{2r-1}{2r}}n^{-1/2})+O(\delta n^{-1/2})=O(\delta^{\frac{2r-1}{2r}}n^{-1/2}).
\]
Combining  $I$ and $II$  yields 
\[
E^{*}\sup_{\delta/2\leq d(\alpha,\alpha_{0}))\leq\delta}\sqrt{n}|(M_{n}-M_{0})(\alpha-\alpha_{0})|\lesssim O(\delta^{\frac{2r-1}{2r}}).
\]

Furthermore,  Lemma $\ref{lem bound 1}$ implies  
\[
\sup_{\delta/2\leq d(\alpha,\alpha_{0}))\leq\delta}P_{\Delta}m_{0}(\cdot,\alpha)-P_{\Delta}m_{0}(\cdot,\alpha_{0})\lesssim-\delta^{2}.
\]

Let $r_{n}=n^{\frac{r}{2r+1}}$. It is easy to check that $r_n$ satisfies $r_{n}^{2}\phi_{n}(\frac{1}{r_{n}})\leq\sqrt{n}$, and
\[
 M_{n}(\hat{\alpha}_{\lambda})\geq M_{n}(\alpha{}_{0})+\lambda[J(\hat{\beta}_{\lambda})-J(\beta_{0})]\geq M_{n}(\alpha_{0})-O_{p}(r_{n}^{-2})
\]
with $\lambda=O(r_{n}^{-2})=O(n^{-\frac{2r}{2r+1}})$.

So far we  have  verified
all the conditions in Theorem 3.4.1 of \cite{Vaart1996} and thus  conclude that 
\[
d(\hat{\alpha},\alpha_{0})=O_p(r_{n}^{-1})=O_p(n^{-\frac{r}{2r+1}}).
\]

For part (iii), recall the projections $a^{*}$ and $g^{*}$ defined in Theorem $\ref{thm:score function}$,
then
\begin{align}
d(\hat{\alpha},\alpha_{0})^{2} & =\mathbb{E}\Delta[\eta_{\hat{\alpha}}(W)-\eta_{\alpha_{0}}(W)]^{2}\nonumber \\
 & =\mathbb{E}\Delta[Z^{'}(\hat{\theta}-\theta_{0})+(\eta_{\hat{\beta}}(X)-\eta_{\beta_{0}}(X))]^{2}\nonumber \\
 & =\mathbb{E}\Delta[(Z-a^{*}(T)-\eta_{g^{*}}(X))^{'}(\hat{\theta}-\theta_{0})+(a^{*}(T)+\eta_{g^{*}}(X))(\hat{\theta}-\theta_{0})\nonumber\\
 & \quad\quad\quad +(\eta_{\hat{\beta}}(X)-\eta_{\beta_{0}}(X))]^{2}\nonumber\\
 & =\mathbb{E}\Delta[(Z-a^{*}(T)-\eta_{g^{*}}(X))^{'}(\hat{\theta}-\theta_{0})]^{2}\nonumber \\
 & \quad\quad\quad +\mathbb{E}\Delta[(a^{*}(T)+\eta_{g^{*}}(X))(\hat{\theta}-\theta_{0})+(\eta_{\hat{\beta}}(X)-\eta_{\beta_{0}}(X))]^{2}.   \label{eq:dXidThetadbeta}
\end{align}
Since $I(\theta)$ is non-singular, it follows that$||\hat{\theta}-\theta_0||^{2}=O_p(n^{-\frac{2r}{2r+1}})$.  This in turn implies 
\[
d(\hat{\beta},\beta_{0})^{2}=O_p(n^{-\frac{2r}{2r+1}}).
\]

\subsection{Proof of Theorem $\ref{thm: asymp normal}$}

Let $u=(t,Z,X(\cdot)).$ For $g\in\mathcal{H}(K)$, define 
\[
s_{n}(u,\alpha)[g]=\eta_{g}(X)-\frac{S_{1n}(t,\alpha)[g]}{S_{0n}(t,\alpha)},\ \ \ \ s(u,\alpha)[g]=\eta_{g}(X)-\frac{S_{1}(t,\alpha)[g]}{S_{0}(t,\alpha)},
\]
and for $Z\in\mathbb{R}^{d}$ and the identify map $I(Z)=Z$, define  
\[
s_{n}(u,\alpha)[Z]=Z-\frac{S_{1n}(t,\alpha)[I]}{S_{0n}(t,\alpha)},\ \ \ \ s(u,\alpha)[Z]=\eta_{g}(X)-\frac{S_{1}(t,\alpha)[I]}{S_{0}(t,\alpha)},
\]
where $S_{1n}(t,\alpha)[I]=\frac{1}{n}\sum_{j=1}^{n}Y_{j}(t)e^{\eta_{\alpha}(W_{j})}Z_{j}$
and $S_{1}(t,\alpha)[I]=\mathbb{E}Y(t)e^{\eta_{\alpha}(W)}Z$. 

By analogy to the score function, we call the derivatives of the
partial likelihood with respect to the parameters the partial score
functions. The partial score function based on the partial likelihood
for $\theta$ is 
\[
i_{n\theta}(\alpha)=P_{\Delta n}s_{n}(\cdot,\alpha)[Z].
\]
The partial score function based on the partial likelihood for $\beta$ in
a direction $g\in\mathcal{H}(K)$ is 
\[
i_{n\beta}(\alpha)[g]=P_{\Delta n}s_{n}(\cdot,\alpha)[g].
\]
Recall that $(\hat{\theta},\hat{\beta})$ is defined to maximize the
penalized partial likelihood, i.e. 
\[
-P_{\Delta n}m_{n}(\cdot,\hat{\theta},\hat{\beta})+\lambda J(\hat{\beta})\leq-P_{\Delta n}m_{n}(\cdot,\theta,\beta)+\lambda J(\beta),
\] 
for all $\theta\in\mathbb{R}^{p}$ and $\beta\in\mathcal{H}(K)$. Since the penalty term is unrelated to $\theta$,  the partial
score function should satisfy
\[
i_{n\theta}(\hat{\alpha})=P_{\Delta n}s_{n}(\cdot,\hat{\alpha})[Z]=0.
\]
On the other hand, the partial score function for $\beta$ satisfies
\[
i_{n\beta}(\hat{\alpha})[g]=P_{\Delta n}s_{n}(\cdot,\alpha)[g]=O(\lambda)=o_{p}(n^{-\frac{1}{2}}),\ \ \ \text{for all }g\in\mathcal{H}(K).
\]
Combining this with Lemma $\ref{lm: normal 2}$ and Lemma $\ref{lm: normal 3}$,
we have 
\[
n^{1/2}P_{\Delta}\{s(\cdot,g_{0})[Z-h^{*}]\}^{\otimes2}(\hat{\theta}-\theta_{0})=-n^{1/2}P_{\Delta n}s_{n}(\cdot,\alpha_{0})[Z-g^{*}]+o_{p}(1).
\]

Let 
\[
M_{i}(t)=\Delta_{i}I\{T_{i}\leq t\}-\int_{0}^{t}Y_{i}(u)\exp(\eta_{\alpha_{0}}(W_{i}))dH_{0}(u),\ \ 1\leq i\leq n.
\]
We can write 
\[
n^{1/2}P_{\Delta n}s_{n}(\cdot,\alpha_{0})[Z-g^{*}]=n^{-1/2}\sum_{i=1}^{n}\int_{0}^{\tau}[Z_{i}-\eta_{h^{*}}(X_{i})-\frac{S_{1n}(t,\alpha_{0})[Z-g^{*}]}{S_{0n}(t,\alpha_{0})}]dM_{i}(t).
\]
Thus 
\begin{align*}
 & n^{1/2}P_{\Delta n}s_{n}(\cdot,\alpha_{0})[Z-g^{*}]-n^{-1/2}\sum_{i=1}^{n}\int_{0}^{\tau}[Z_{i}-\eta_{h^{*}}(X_{i})-\frac{S_{1}(t,\alpha_{0})[Z-g^{*}]}{S_{0}(t,\alpha_{0})}]dM_{i}(t)\\
 & =n^{-1/2}\sum_{i=1}^{n}\int_{0}^{\tau}[\frac{S_{1}(t,\alpha_{0})[Z-g^{*}]}{S_{0}(t,\alpha_{0})}-\frac{S_{1n}(t,\alpha_{0})[Z-g^{*}]}{S_{0n}(t,\alpha_{0})}]dM_{i}(t).
\end{align*}
Because 
\[
n^{-1}\sum_{i=1}^{n}\int_{0}^{\tau}[\frac{S_{1}(t,\alpha_{0})[Z-g^{*}]}{S_{0}(t,\alpha_{0})}-\frac{S_{1n}(t,\alpha_{0})[Z-g^{*}]}{S_{0n}(t,\alpha_{0})}]Y_{i}(t)\exp[\eta_{\alpha_{0}}(W_{i})]dH_{i}(t)\stackrel{P}{\rightarrow}0,
\]
by Lenglart's inequality, as stated in Theorem 3.4.1 and Corollary
3.4.1 of \cite{Fleming1991}, we have 
\begin{eqnarray*}
& & n^{1/2}P_{\Delta n}s_{n}(\cdot,\alpha_{0})[Z-g^{*}]\\
& & \quad \quad\quad=n^{-1/2}\sum_{i=1}^{n}\int_{0}^{\tau}[Z_{i}-\eta_{h^{*}}(X_{i})-\frac{S_{1}(t,\alpha_{0})[Z-g^{*}]}{S_{0}(t,\alpha_{0})}]dM_{i}(t)+o_{p}(1).
\end{eqnarray*}
Recall that 
\[
\frac{S_{1}(t,\alpha_{0})[Z-g^{*}]}{S_{0}(t,\alpha_{0})}=E[Z-\eta_{g^{*}}(W)]|T=t,\Delta=1]=a^{*}(t).
\]
By the definition of the efficient score function $l_{\theta}^{*}$,
we have 
\[
n^{1/2}P_{\Delta n}s_{n}(\cdot,\alpha_{0})[Z-g^{*}]=n^{-1/2}\sum_{i=1}^{n}l_{\theta}^{*}(T_{i},\Delta_{i},W_{i})+o_{p}(1)\stackrel{d}{\rightarrow\mathcal{N}(0,I(\theta_{0}))}.
\]

\subsection{Proof of Theorem $\ref{thm: lower bound}$}\label{sub: pf of lower bound}
To get the minmax lower bound,
it  suffices to show that, when the true baseline hazard function $h_0$ and  the true $\theta_0$ are fixed and known, for a subset  $\mathcal{H}^{*}$ of $\mathcal{H}(K),$  
\begin{equation}
\lim_{a\rightarrow0}\lim_{n\rightarrow\infty}\inf_{\hat{\beta}}\sup_{\beta_{0}\in\mathcal{H}^{*}}\mathbb P_{h_0,\theta_0,\beta_0}\{d(\hat{\beta},\beta_{0})\geq an^{-\frac{r}{2r+1}}\}=1.\label{eq: lower bound with subset}
\end{equation}

If we can find a subset $\{\beta^{(0)},\ldots,\beta^{(N)}\}\subset\mathcal{H}^{*}$
with $N$ increasing with $n$, such that for some positive constant
$c$ and all $0\leq i<j\leq N$,

\begin{equation}
d^{2}(\beta^{(i)},\beta^{(j)})\geq c\gamma^{\frac{2r}{2r+1}}n^{-\frac{2r}{2r+1}},\label{eq:bound 1}
\end{equation}
 and
\begin{eqnarray}
\frac{1}{N}\sum_{j=1}^{N}KL(P_{j},P_{0}) & \leq & \gamma\log N,\label{eq: bound 2}
\end{eqnarray}
then we can conclude, according to \cite{Tsybakov2009} (Theorem 2.5 on page 99), 
that
\[
\inf_{\hat{\beta}}\sup_{\beta\in\mathcal{H}^{*}}\mathbb{P}(d^{2}(\beta^{(i)},\beta^{(j)})\geq c\gamma^{\frac{2r}{2r+1}}n^{-\frac{2r}{2r+1}})\geq\frac{\sqrt{N}}{1+\sqrt{N}}(1-2\gamma-\sqrt{\frac{2\gamma}{\log N}}),
\]
which yields 
\[
\lim_{a\rightarrow0}\lim_{n\rightarrow\infty}\inf_{\hat{\beta}}\sup_{\beta_0\in\mathcal{H}^{*}}\mathbb{P}(d(\beta^{(i)},\beta^{(j)})\geq an^{-\frac{r}{2r+1}})\geq1.
\]
Hence Theorem $\ref{thm: lower bound}$ will be proved.

Next,  we are going to construct the set $\mathcal{H}^{*}$ and the subset
$\{\beta^{(0)},\ldots,\beta^{(N)}\}\subset\mathcal{H}^{*}$,  and then
show that both $(\ref{eq:bound 1})$ and $(\ref{eq: bound 2})$ are satisfied.

Consider the  function space
\begin{equation}
\mathcal{H}^{*}=\{\beta=\sum_{k=M+1}^{2M}b_{k}M^{-1/2}L_{K^{1/2}}\varphi_{k}:(b_{M+1},\ldots b_{2M})\in\{0,1\}^{M}\}, \label{eq: subset for lower bound}
\end{equation}
where $\{\varphi_{k}:k\geq1\}$  are the orthonomal eigenfunctions
of $T(s,t)=K^{1/2}C_{\Delta}K^{1/2}(s,t)$ and $M$ is some large
number to be decided later. 

For any $\beta\in\mathcal{H}^{*}$, observe that 
\begin{align*}
||\beta||_{K}^{2} & =||\sum_{k=M+1}^{2M}b_{k}M^{-1/2}L_{K^{1/2}}\varphi_{k}||_{K}^{2}\\
 & =\sum_{k=M+1}^{2M}b_{k}^{2}M^{-1}||L_{K^{1/2}}\varphi_{k}||_{K}^{2}\\
 & \leq\sum_{k=M+1}^{2M}M^{-1}||L_{K^{1/2}}\varphi_{k}||_{K}^{2}\\
 & =1,
\end{align*}
which follows from the fact that 
\begin{align*}
<L_{K^{1/2}}\varphi_{k},L_{K^{1/2}}\varphi_{l}>_{K} & =<L_{K}\varphi_{k},\varphi_{l}>_{K}=<\varphi_{k},\varphi_{l}>_{{L}_{2}}=\delta_{kl}.
\end{align*}
Therefore $\mathcal{H}^{*}\subset\mathcal{H}(K)=\{\beta:||\beta||_{k}<\infty\}$.

The Varshamov-Gilbert bound shows that for any $M\geq8$, there exists
a set $\mathcal{B}=\{b^{(0)},b^{(1)},\ldots,b^{(N)}\}\subset\{0,1\}^{M}$such
that
\begin{enumerate}
\item $b^{(0)}=(0,\ldots,0)'$;
\item $H(b,b')>M/8$ for any $b\neq b'\in\mathcal{B}$, where $H(\cdot,\cdot)=\frac{1}{4}\sum_{i=1}^{M}(b_{i}-b_{i}')^{2}$
is the Hamming distance;
\item $N\geq2^{M/8}$.
\end{enumerate}
The subset $\{\beta^{(0)},\ldots,\beta^{(N)}\}\subset\mathcal{H}^{*}$
is chosen as $\beta^{(i)}=\sum_{k=M+1}^{2M}b_{k-M}^{(i)}M^{-1/2}L_{K^{1/2}}\varphi_{k},$ $i=0,\ldots N$. 

For any $0\leq i<j\leq N$, observe that 
\begin{eqnarray*}
d^{2}(\beta^{(i)},\beta^{(j)}) & = & \mathbb{E}\Delta\big(\eta_{\beta^{(i)}}(X)-\eta_{\beta^{(j)}}(X)\big)^{2}\\
 & = & ||L_{C_{\Delta}^{1/2}}\sum_{k=M+1}^{2M}(b_{k-M}^{(i)}-b_{k-M}^{(j)})M^{-1/2}L_{K^{1/2}}\varphi_{k}||_{{L}_{2}}^{2}\\
 & = & \sum_{k=M+1}^{2M}(b_{k-M}^{(i)}-b_{k-M}^{(j)})^{2}M^{-1}||L_{C_{\Delta}^{1/2}}L_{K^{1/2}}\varphi_{k}||_{{L}_{2}}^{2}\\
 & = & \sum_{k=M+1}^{2M}(b_{k-M}^{(i)}-b_{k-M}^{(j)})^{2}M^{-1}s_{k}.
\end{eqnarray*}
On one hand, we have
\begin{eqnarray*}
d^{2}(\beta^{(i)},\beta^{(j)}) & = & \sum_{k=M+1}^{2M}(b_{k-M}^{(i)}-b_{k-M}^{(j)})^{2}M^{-1}s_{k}\\
 & \geq & s_{2M}M^{-1}\sum_{k=1}^{M}(b_{k}^{(i)}-b_{k}^{(j)})^{2}\\
 & = & 4s_{2M}M^{-1}H(b^{(i)},b^{(j)})\\
 & \geq & s_{2M}/2.
\end{eqnarray*}
On the other hand, we have
\begin{eqnarray*}
d^{2}(\beta^{(i)},\beta^{(j)}) & = & \sum_{k=M+1}^{2M}(b_{k-M}^{(i)}-b_{k-M}^{(j)})^{2}M^{-1}s_{k}\\
 & \leq & s_{M}M^{-1}\sum_{k=1}^{M}(b_{k}^{(i)}-b_{k}^{(j)})^{2}\\
 & \leq & s_{M}.
\end{eqnarray*}
So altogether,
\begin{equation}
s_{2M}/2\leq d^{2}(\beta^{(i)},\beta^{(j)})\leq s_{M}.\label{eq: d sM}
\end{equation}

Let $P_{j},\ j=1,\ldots N,$ be the likelihood function with data $\{(T_{i},\Delta_{i},W_{i}(s)),i=1,\ldots n\}$
and $\beta^{(j)}$, i.e
\[
P_{j}=\prod_{i=1}^{n}\big[f_{T^{u}|W}(T_{i})S_{T^{c}|W}(T_{i})\big]^{\Delta_{i}}\cdot\big[f_{T^{c}|W}(T_{i})S_{T^{u}|W}(T_{i})\big]^{1-\Delta_{i}}.
\]
Let $c_{T^{c}}=\prod_{i=1}^{n}\big[S_{T^{c}|W}(T_{i})\big]^{\Delta_{i}}\big[S_{T^{u}|W}(T_{i})\big]^{1-\Delta_{i}}$,
which does not depend on $\beta^{(j)}$, then
\[
P_{j}=c_{T^{c}}\prod_{i=1}^{n}\big[h_{0}(T_{i})\exp(\theta_{0}'Z_{i}+\eta_{\beta^{(j)}}(X_{i}))\big]^{\Delta_{i}}\cdot\exp\{-H_{0}(T_{i})\cdot e^{\theta_{0}'Z_{i}+\eta_{\beta^{(j)}}(X_{i})}\}.
\]
We calculte the Kullback-Leibler distance between
$P_{j}$ and $P_{0}$ as
\begin{align*}
KL(P_{j},P_{0}) & =\mathbb{E}_{P_{j}}\log\frac{P_{j}}{P_{0}}\\
 & =\mathbb{E}_{P_{j}}\Big\{\Delta_{i}\sum_{i=1}^{n}\{\eta_{\beta^{(j)}-\beta^{(0)}}(X_{i})\}+\sum_{i=1}^{n}H_{0}(T_{i})\, e^{\theta_{0}'Z_{i}}[\exp(\eta_{\beta^{(0)}}(X_{i}))\\
 & \quad\quad\quad-\exp(\eta_{\beta^{(j)}}(X_{i}))]\Big\}\\
 & =n\mathbb{E}_{P_{j}}\Delta[\eta_{\beta^{(j)}-\beta^{(0)}}(X)]+n\mathbb{E}_{P_{j}}H_{0}(T)\, e^{\theta_{0}'Z}[\exp(\eta_{\beta^{(0)}}(X))\\
 &\quad\quad\quad-\exp(\eta_{\beta^{(j)}}(X))]\\
 & =n\mathbb{E}_{P_{j}}^{W}\mathbb{E}_{P_{j}}^{T,\Delta}\{H_{0}(T)\,|W\}\, e^{\theta_{0}'Z}[\exp(\eta_{\beta^{(0)}}(X))-\exp(\eta_{\beta^{(j)}}(X))]\Big],
\end{align*}
where
\begin{align*}
\mathbb{E}_{p_{j}}^{T,\Delta}(H_{0}(T)\,|W) & \text{=}\mathbb{E}^{T^{c}}\big\{\mathbb{E}_{p_{j}}^{T,\Delta}(H_{0}(T)\,|T^{c},W)\big|W\big\}\\
 & =\mathbb{E}^{T^{c}}\big\{\int_{0}^{T^{c}}H_{0}(t)f_{T^{u}|W}(t)dt+H_{0}(T^{c})\mathbb{P}(T^{u}>T^{c}|T^{c},W)\big|W)\big\}, 
\end{align*}
\begin{align*}
 & \int_{0}^{T^{c}}H_{0}(t)f_{T^{u}|W}(t)dt\\
 & =\int_{0}^{T^{c}}H_{0}(t)\cdot h_{0}(t)\exp[\theta_{0}'Z+\eta_{\beta^{(j)}}(X)]\exp\{-H_{0}(t)e^{\theta_{0}'Z+\eta_{\beta^{(j)}}(X)}\}dt\\
 & =e^{-\theta_{0}'Z-\eta_{\beta^{(j)}}(X)}\int_{0}^{T^{c}}e^{\theta_{0}'Z+\eta_{\beta^{(j)}}(X)}H_{0}(T)\exp\{-H_{0}(T)\cdot e^{\theta_{0}'Z+\eta_{\beta^{(j)}}(X)}\}\\
 &\quad\quad\quad de^{\theta_{0}'Z+\eta_{\beta^{(j)}}(X)}H_{0}(T)\\
 & =\exp(-\theta_{0}'Z+\eta_{\beta^{(j)}}(X))\int_{0}^{a}ue^{-u}du\ \ \Big|_{a=e^{\theta_{0}'Z+\eta_{\beta^{(j)}}(X)}H_{0}(T^{c})}\\
 & =\exp(-\theta_{0}'Z-\eta_{\beta^{(j)}}(X))[1-e^{-a}-a\, e^{-a}]\Big|_{a=e^{\theta_{0}'Z+\eta_{\beta^{(j)}}(X)}H_{0}(T^{c})},
\end{align*}
and
\begin{align*}
\mathbb{P}(T^{u}>T^{c}|T^{c},W) & =S_{T^{u}|W}(T^{c})\\
 & =\exp\{-H_{0}(T^{c})e^{\theta_{0}'Z+\eta_{\beta^{(j)}}(X)}\}.
\end{align*}
Therefore 
\begin{align*}
& \mathbb{E}_{p_{j}}^{T,\Delta}(H_{0}(T)\,|T^{c},W)\\
&=e^{-\theta_{0}'Z-\eta_{\beta^{(j)}}(X)}[1-\exp\{-H_{0}(T^{c})e^{\theta_{0}'Z+\eta_{\beta^{(j)}}(X)}\}]-H_{0}(T^{c})\cdot \\
 & \quad\quad\quad\exp\{-H_{0}(T^{c})e^{\theta_{0}'Z+\eta_{\beta^{(j)}}(X)}\} +H_{0}(T^{c})\exp\{-H_{0}(T^{c})e^{\theta_{0}'Z+\eta_{\beta^{(j)}}(X)}\}\\
 & =e^{-\theta_{0}'Z-\eta_{\beta^{(j)}}(X)}[1-\exp\{-H_{0}(T^{c})e^{\theta_{0}'Z+\eta_{\beta^{(j)}}(X)}\}]\\
 & =e^{-\theta_{0}'Z-\eta_{\beta^{(j)}}(X)}[F_{T^{u}|W}(T^{c})]\\
 & =e^{-\theta_{0}'Z-\eta_{\beta^{(j)}}(X)}\mathbb{P}(T^{u}\leq T^{c}|T^{c},W),
\end{align*}
and further 
\begin{align*}
\mathbb{E}_{p_{j}}^{T,\Delta}(H_{0}(T)\,|W) & \text{=}\mathbb{E}^{T^{c}}\big\{\mathbb{E}_{p_{j}}^{T,\Delta}(H_{0}(T)\,|T^{c},W)\big|W\big\}\\
 & =e^{-\theta_{0}'Z-\eta_{\beta^{(j)}}(X)}\mathbb{P}(T^{u}\leq T^{c}|W)\\
 & =e^{-\theta_{0}'Z-\eta_{\beta^{(j)}}(X)}\mathbb{E}[\Delta|W].
\end{align*}
Then the KL distance becomes

\begin{align*}
KL(P_{j},P_{0}) & =n\mathbb{E}_{P_{j}}^{W}\mathbb{E}[\Delta|W]e^{-\theta_{0}'Z-\eta_{\beta^{(j)}}(X)}\, e^{\theta_{0}'Z}[\exp(\eta_{\beta^{(0)}}(X))-\exp(\eta_{\beta^{(j)}}(X))]\Big]\\
 & =n\mathbb{E}_{P_{j}}^{W,\Delta}\Delta[\exp(\eta_{\beta^{(0)}}(X)-\eta_{\beta^{(j)}}(X))-1\Big]\\
 & =n\mathbb{E}_{P_{j}}^{W,\Delta}[\frac{1}{2}\Delta(\eta_{\beta^{(0)}}(X)-\eta_{\beta^{(j)}}(X))^{2}+o(\Delta(\eta_{\beta^{(0)}}(X)-\eta_{\beta^{(j)}}(X))^{2})]\\
 & \leq n\mathbb{E}_{P_{j}}^{X}[\frac{1}{2}(\eta_{\beta^{(0)}}(X)-\eta_{\beta^{(j)}}(X))^{2}+o((\eta_{\beta^{(0)}}(X)-\eta_{\beta^{(j)}}(X))^{2})]\\
 & \lesssim nd^{2}(\beta^{(j)},\beta^{(0)})\\
 & \lesssim ns_{M}.
\end{align*}

Therefore for some positive constant $c_{1}$, 
\[
KL(P_{j},P_{0})\leq c_{1}nM^{-2r}.
\]
By taking $M$ to be the smallest integer greater than $c_{2}\gamma^{-\frac{1}{2r+1}}n^{\frac{1}{2r+1}}$
with $c_{2}=(c_{1}\cdot8\log2)^{1/(1+2r)}$, we verified $(\ref{eq: bound 2})$
that 
\begin{eqnarray*}
\frac{1}{N}\sum_{j=1}^{N}KL(P_{j},P_{0}) & \leq & \gamma\log N.
\end{eqnarray*}

Meanwhile, since $d^{2}(\beta^{(i)},\beta^{(j)})\geq s_{2M}/2$  and
$s_{2M}\asymp(2M)^{-2r}$, condition  $(\ref{eq:bound 1})$
is verified  by plugging in $M$. 

\section{Appendix}

\subsection{Derivation of $GCV(\lambda)$\label{sub: GCV}}

Recall that given the observations $\{(T_{i},\Delta_{i},$ $W_{i})\}_{i=1}^{n}$,
$\hat{\beta}_{\lambda}$ can be written in the form of 
\[
\hat{\beta}_{\lambda}(t)=\sum_{k=1}^{m}d_{k}\xi_{k}(t)+\sum_{i=1}^{n}c_{i}\int_{0}^{1}X_{i}(s)K_{1}(s,t)ds. 
\]
For simplicity, let $\xi_{k+j}(t)=\int_{0}^{1}X_{j}(s)K_{1}(s,t)ds,\, j=1,\ldots n$,
then write $\beta(t)=\sum_{k=1}^{m+n}c_{k}^{(\beta)}\xi_{k}(t)$. 
In this way, 
\[
\eta_{\alpha}(W_i)=\sum_{k=1}^p \theta_k Z_{ik}+\sum_{k=1}^{m+n} c_{k}^{(\beta)} \int X_i (t)\xi_{k}(t)dt. 
\]
Let $S^{(\beta)}$ be an  $n\times(m+n)$ matrix with the $(i,j)th$ entry defined as $S_{ij}^{(\beta)}=\int X_{i}(s)\xi_{j}(s)ds$, and  $\bold{Z}=(Z_1,\cdots,Z_n)_{n\times p} $. Denote $S=(\bold{Z},S^{(\beta)})$, a $ n\times (p+m+n)$ matrix, and $(\eta_{\alpha}(W_{1}),\ldots\eta_{\alpha}(W_{n}))^{T}=S\cdot c$ with $c=(c_{1},\ldots, c_{p+n+m})^{T}=(\theta_1,\ldots \theta_p,c_1^{(\beta)},\ldots,c_{m+n}^{(\beta )})^T$.

Since $\xi_{1},\ldots\xi_{m}$ are the bases of the null space with the  semi-norm
$J(\cdot)$,  we can write $J$ as $J(\beta)=c^{T}Qc$,
with $Q$ a $(p+m+n)\times(p+m+n)$ diagonal block matrix whose  non-zero
entries only occur in the $n\times n$ submatrix $(Q_{i,j})_{i,j=p+m+1}^{p+m+n}$. 

Let $\Delta=(\Delta_{1},\ldots,\Delta_{n})^{T}$ and $Y_{j}(t)=I\{t\geq T_{j}\}$.
Under the above expressions, we can write the penalized partial likelihood as  a function of the coefficient $c$:
\[
A_{\lambda}(c)=-\Delta'S\cdot c/n+\frac{1}{n}\sum_{i=1}^{n}\Delta_{i}\log\{\sum_{j=1}^{n}Y_{j}(T_{i})e^{S_{j\cdot}c}\}+\lambda c^{T}Qc,
\]
where $S_{j\cdot}$ is the $j^{th}$ row of $S$.

For any $\alpha\in \mathbb{R}^p\times \mathcal H(K)$, functions $f,g\in\mathcal H(K)$, and $z,z^*\in \mathbb R^{n\times 1}$ define
\[
\mu_{\alpha}(f|t)=\frac{\sum_{j=1}^{n}Y_{j}(t)e^{\eta_{\alpha}(W_j)}\eta_f (X_j)}{\sum_{j=1}^{n} Y_{j}(t)e^{\eta_{\alpha}(W_j)}},
\ \ \ \ \ \ \ 
\mu_{\alpha}(z|t)=\frac{\sum_{j=1}^{n}Y_{j}(t)e^{\eta_{\alpha}(W_j)}z_j}{\sum_{j=1}^{n} Y_{j}(t)e^{\eta_{\alpha}(W_j)}},
\]
and 
\[
\mu_{\alpha}(f,\, g|t)=\frac{\sum_{j=1}^{n}Y_{j}(t)e^{\eta_{\alpha}(W_j)} \eta_f (X_j)\cdot \eta_g (X_j)}{\sum_{j=1}^{n}Y_{j}(t)e^{\eta_{\alpha}(W_j)}},
\]
\[
\mu_{\alpha}(z,\, z^*|t)=\frac{\sum_{j=1}^{n}Y_{j}(t)e^{\eta_{\alpha}(W_j)} z_j\,z^*_j}{\sum_{j=1}^{n}Y_{j}(t)e^{\eta_{\alpha}(W_j)}},
\]
\[
\mu_{\alpha}(f,\, z|t)=\frac{\sum_{j=1}^{n}Y_{j}(t)e^{\eta_{\alpha}(W_j)} \eta_f (X_j)\,z_j}{\sum_{j=1}^{n}Y_{j}(t)e^{\eta_{\alpha}(W_j)}}.
\]
Define  $\mu_{\alpha}(g)=\frac{1}{n}\sum_{i=1}^{n}\mu_{\alpha}(g|T_{i})$,
$V_{\alpha}(f,g|t)=\mu_{\alpha}(f,\, g|t)-\mu_{\alpha}(f|t)\mu_{\alpha}(g|t)$, and $V_{\alpha}(f,g)=\frac{1}{n}\sum_{i=1}^{n}V_{\alpha}(f,g|T_{i})$, and define by analogy  $\mu_{\alpha}(z)$, $V_{\alpha}(z,z^*|t)$, $V_{\alpha}(f,z|t)$, $V_{\alpha}(z,z^*)$, and $V_{\alpha}(f,z)$.
Now take the derivative of $A_{\lambda}(c)$ at $\tilde{\alpha}=S\cdot\tilde{c}$
with respect to $c$, we have 
\[
\frac{\partial A_{\lambda}(c)}{\partial c}\big|_{\tilde{\alpha}}=-S^T\Delta/n+\mu_{\tilde{\alpha}}(\varsigma)+2\lambda Q\tilde{c},
\]
and 
\[
\frac{\partial^{2}A_{\lambda}(c)}{\partial c^{2}}\big|_{\tilde{\alpha}}=V_{\tilde{\alpha}}(\varsigma,\varsigma^{T})+2\lambda Q,
\]
where $\varsigma=(Z_{\cdot 1},\ldots,Z_{\cdot p},\xi_{1}(s),\ldots,\xi_{m+n}(s))^{T}$.  To obtain the minimum of $A_{\lambda}(c)$, we apply the Newton-Raphson algorithm to $\partial A_{\lambda}(c)/\partial c$. That is,
\[
[V_{\tilde{\alpha}}(\varsigma,\varsigma^{T})+2\lambda Q](c-\tilde{c})=S^T\Delta/n-\mu_{\tilde{\alpha}}(\varsigma)-2\lambda Q\tilde{c}.
\]
To simplify the notations, let $H=[V_{\tilde{\alpha}}(\varsigma,\varsigma^{T})+2\lambda Q]$
and $h=-\mu_{\tilde{\alpha}}(\varsigma)+V_{\tilde{\alpha}}(\varsigma,\varsigma^{T})\tilde{c}$, 
so $\hat{c}\thickapprox H^{-1}(S'\Delta/n+h)$ and 
\[
\hat{c}^{[i]}\thickapprox H^{-1}(\frac{S^T\Delta-\Delta_{i}S_{i\cdot}^T}{n-1}+h)=\hat{c}-\Delta_{i}\cdot\frac{H^{-1}S_{i\cdot}^T}{n-1}+\frac{H^{-1}S'\Delta}{n(n-1)}.
\]
Then the first term of $\widehat{RKL}$ becomes 
\[
\sum_{i=1}^{n}\eta_{\tilde{\alpha}_{\lambda}}^{[i]}(W_{i})=\sum_{i=1}^{n}\eta_{\tilde{\alpha}_{\lambda}}(W_{i})-\sum_{i=1}^{n}[\Delta_{i}\cdot\frac{S_{i\cdot} H^{-1}S_{i\cdot}^T}{n-1}+\frac{S_{i\cdot} H^{-1}S'\Delta}{n(n-1)}].
\]
Simplifying this leads to  
\[
\sum_{i=1}^{n}\eta_{\tilde{\alpha}_{\lambda}}^{[i]}(W_{i})=\sum_{i=1}^{n}\eta_{\tilde{\alpha}_{\lambda}}(W_{i})-\frac{1}{(n-1)}tr[(SH^{-1}S)(diag\Delta-\Delta\boldsymbol{1}'/n)],
\]
where $diag\Delta$ is an $n\times n$ diagonal matrix with diagonal
entries $\Delta_{1},\ldots,\Delta_{n}$. Plugging this back to $\widehat{RKL}$,
then $GCV(\lambda)$ is obtained.

If the \textcolor{black}{efficient estimator $\beta_{\lambda}^{*}$ is used instead,
the derivation and therefore the main result remain the same by adjusting the definition of $\xi$ and $S^{(\beta)}$ accordingly.}

\subsection{Proofs of Lemmas}


\begin{lem}
\label{lem: dH}Following the former notations, for $0\leq s\leq1$,
let 
\[
g(t,s)=\frac{S_{1}(t,\alpha_{0}+s\widetilde{\alpha})[\alpha^{*}]}{S_{0}(t,\alpha_{0}+s\widetilde{\alpha})}.
\]
Denote $R_{s}(t)=Y(t)\exp(\eta_{\alpha_{0}}+s\eta_{\widetilde{\alpha}})/S_{0}(t,\alpha_{0}+s\widetilde{\alpha})$.
We have
\begin{eqnarray*}
\frac{\partial}{\partial s}g(t,s) & = & \mathbb{E}[R_{s}(t)\eta_{\widetilde{\alpha}}\eta_{\alpha^{*}}]-\mathbb{E}[R_{s}(t)\eta_{\widetilde{\alpha}}]\mathbb{E}[R_{s}(t)\eta_{\alpha^{*}}]\\
 & = & \mathbb{E}\big\{ R_{s}(t)\big(\eta_{\widetilde{\alpha}}-\mathbb{E}[R_{s}(t)\eta_{\widetilde{\alpha}}]\big)\big(\eta_{\alpha^{*}}-\mathbb{E}[R_{s}(t)\eta_{\alpha^{*}}]\big),
\end{eqnarray*}
and 
\begin{align*}
\frac{\partial^{2}}{\partial s^{2}}g(t,s) & =\mathbb{E}[R_{s}(t)\eta_{\alpha^{*}}\eta_{\widetilde{\alpha}}^{2}]-2\mathbb{E}[R_{s}(t)\eta_{\widetilde{\alpha}}]\mathbb{E}[R_{s}(t)\eta_{\alpha^{*}}\eta_{\widetilde{\alpha}}] \\
 & \ \ \ \ \ \ \ \ -\mathbb{E}[R_{s}(t)\eta_{\alpha^{*}}]\mathbb{E}[R_{s}(t)\eta_{\widetilde{\alpha}}^{2}]+2\mathbb{E}[R_{s}(t)\eta_{\alpha^{*}}]\mathbb{E}[R_{s}(t)\eta_{\widetilde{\alpha}}]^{2}.
\end{align*}

\end{lem}
\begin{proof}
The lemma follows by direct calculation.   
\end{proof}

\begin{lem}
\label{lem bound 1}Let $\alpha_{0}$ be the true coefficients. Under
assumption A(1)-A(4), we have

\[
P_{\Delta}m_{0}(\cdot,\alpha)-P_{\Delta}m_{0}(\cdot,\alpha_{0})\asymp-d^{2}(\alpha,\alpha_{0}).
\]
\end{lem}
\begin{proof}
Observe that
\begin{eqnarray*}
& &P_{\Delta}m_{0}(\cdot,\alpha)-P_{\Delta}m_{0}(\cdot,\alpha_{0})\\
 & & = P_{\Delta}(m_{0}(\cdot,\alpha)-m_{0}(\cdot,\alpha_{0}))\\
 & & = P_{\Delta}\{\eta_{\alpha-\alpha_{0}}(W)-\log S_{0}(\cdot,\alpha)+\log S_{0}(\cdot,\alpha_{0})\}1_{\{0\leq T\leq\tau\}}\\
 & & = -P_{\Delta}\{\log S_{0}(\cdot,\alpha)-\log S_{0}(\cdot,\alpha_{0})\}1_{\{0\leq T\leq\tau\}}.
\end{eqnarray*}
Let $\widetilde{\alpha}=(\theta-\theta_{0},\beta-\beta_{0})$ and $G(t,s)=\log(S_{0}(t,\alpha_{0}+s\widetilde{\alpha}))$, 
then
\[
P_{\Delta}m_{0}(\cdot,\alpha)-P_{\Delta}m_{0}(\cdot,\alpha_{0})=-P_{\Delta}(G(\cdot,1)-G(\cdot,0))1_{\{0\leq T\leq\tau\}}.
\]
For fixed t, take the derivative of $G(t,s)$ with respect to $s$, we
have 
\[
\frac{\partial}{\partial s}G(t,s)=\frac{S_{1}(t,\alpha_{0}+s\widetilde{\alpha})[\widetilde{\alpha}]}{S_{0}(t,\alpha_{0}+s\widetilde{\alpha})}\stackrel{def}{=}g(t,s).
\]
Noting that $P_{\Delta}\frac{\partial}{\partial s}G(\cdot,0)=P_{\Delta}g(\cdot,0)=0$, then  lemma
$\ref{lem: dH}$ implies, 
\[
\frac{\partial^{2}}{\partial s^{2}}G(t,s)=\frac{\partial}{\partial s}g(t,s)=\mathbb{E}[R_{s}(t)\eta_{\widetilde{\alpha}}^{2}]-\big(\mathbb{E}[R_{s}(t)\eta_{\widetilde{\alpha}}]\big)^{2}=\mathbb{E}R_{s}(t)\big(\eta_{\widetilde{\alpha}}-\mathbb{E}[R_{s}(t)\eta_{\widetilde{\alpha}}]\big)^{2},
\]
where $R_{s}(t)=\frac{Y(t)e^{\eta_{\alpha_{0}}+s\eta_{_{\widetilde{\alpha}}}}}{S_{0}(t,\eta_{\alpha_{0}}+s\eta_{\widetilde{\alpha}})}$.
Therefore for some $\gamma\in[0,1]$,

\begin{eqnarray*}
G(t,1)-G(t,0) & = & G'_{s}(t,0)+\frac{1}{2}G''_{s}(t,\gamma)\\
 & = &g(t,0)+\frac{1}{2} \mathbb{E}R_{\gamma}(t)\big(\eta_{\widetilde{\alpha}}-\mathbb{E}[R_{\gamma}(t)\eta_{\widetilde{\alpha}}]\big)^{2}\\
 & = & g(t,0)+\frac{1}{2} \mathbb{E}^{W}\mathbb{E}\big(R_{\gamma}(t)|W\big)\big(\eta_{\widetilde{\alpha}}-\mathbb{E}[R_{\gamma}(t)\eta_{\widetilde{\alpha}}]\big)^{2}.
\end{eqnarray*}
By the definition of $R_{s}(t)$,
\[
\mathbb{E}\big(R_{\gamma}(t)|W\big)=P(T\geq t|W)\exp(\eta_{\alpha_{0}+\gamma\widetilde{\alpha}}(W))/S_{0}(t,\eta_{\alpha_{0}+\gamma\widetilde{\alpha}}).
\]
By the assumptions and  for $t\in[0,\tau]$, there exists constants $c_{1}>c_{2}>0$
not depending on $t$, such that
\[
c_{2}\leq\mathbb{E}[R_{\gamma}(t)|W]\leq c_{1}.
\]
On one hand,
\begin{eqnarray*}
& & G(t,1)-G(t,0)\\ 
& & \geq   g(t,0)+\frac{1}{2}c_{2}\mathbb{E}\big(\eta_{\widetilde{\alpha}}-\mathbb{E}[R_{\gamma}(t)\eta_{\widetilde{\alpha}}]\big)^{2}\geq  g(t,0)+\frac{1}{2}c_{2}\mathbb{E}\Delta\big(\eta_{\widetilde{\alpha}}-\mathbb{E}[R_{\gamma}(t)\eta_{\widetilde{\alpha}}]\big)^{2}\\
 & & =   g(t,0)+\frac{1}{2}c_{2}\mathbb{E}\Delta\eta_{\widetilde{\alpha}}^{2}-2\mathbb{E}\Delta\eta_{\widetilde{\alpha}}\mathbb{E}[R_{\gamma}(t)\eta_{h}]+\mathbb{E}[R_{\gamma}(t)\eta_{\widetilde{\alpha}}]^{2}\\
  & & \geq   g(t,0)+\frac{1}{2}c_{2}d^{2}(\alpha,\alpha_{0}),
\end{eqnarray*}
which follows from the fact that $E\Delta\eta_{\widetilde{\alpha}}=0.$
So

\begin{eqnarray}
P_{\Delta}m_{0}(\cdot,\alpha)-P_{\Delta}m_{0}(\cdot,\alpha_{0}) & = & -P_{\Delta}\{G(\cdot,1)-G(\cdot,0)\}1_{\{0\leq T\leq\tau\}}\nonumber \\
 & \leq & -P_{\Delta}d^{2}(\alpha,\alpha_{0})1_{\{0\leq T\leq\tau\}}\nonumber \\
 & \lesssim & -d^{2}(\alpha,\alpha_{0}).  \label{eq: lemma2 1}
\end{eqnarray}

On the other hand,
\begin{eqnarray*}
G(t,1)-G(t,0) & \leq & g(t,0)+ \frac{1}{2}c_{1}\mathbb{E}\big(\eta_{\widetilde{\alpha}}^{2}-E[R_{\gamma}(t)\eta_{\widetilde{\alpha}}]\big)^{2}\\
&\leq&  g(t,0)+c_{1}\{E\eta_{\widetilde{\alpha}}^{2}+(E[R_{\gamma}(t)\eta_{\widetilde{\alpha}}]){}^{2}\}.
\end{eqnarray*}
Since $(E[R_{\gamma}(t)\eta_{\widetilde{\alpha}}]){}^{2}=(E^{W}E[R_{\gamma}(t)|W]^{2}\cdot\eta_{\widetilde{\alpha}}^{2})\leq c_{1}^{2}\epsilon^{-1}E\Delta\eta_{\widetilde{\alpha}}^{2}$,  we arrive at  
\begin{eqnarray}
P_{\Delta}m_{0}(\cdot,\alpha)-P_{\Delta}m_{0}(\cdot,\alpha_{0}) & = & -P_{\Delta}\{G(\cdot,1)-G(\cdot,0)\}1_{\{0\leq T\leq\tau\}}\nonumber \\
 & \gtrsim & -P_{\Delta}d^{2}(\alpha,\alpha_{0})1_{\{0\leq T\leq\tau\}}\nonumber \\
 & \gtrsim & -d^{2}(\alpha,\alpha_{0}).  \label{eq: lemma2 1-1}
\end{eqnarray}

Combining ($\ref{eq: lemma2 1}$) and ($\ref{eq: lemma2 1-1}$) we have
\[
P_{\Delta}m_{0}(\cdot,\alpha)-P_{\Delta}m_{0}(\cdot,\alpha_{0})\asymp-d^{2}(\alpha,\alpha_{0}). 
\]
\end{proof}

\begin{lem}
\label{lem:consist 1}$\mathcal{F}_{1}=\{\Delta m_{0}(t,W,\alpha):\alpha\in\mathbb{R}_{M}\times\mathcal{H}^{M}\}$
and $\mathcal{F}_{2}=\{Y(t)e^{\eta_{\alpha}(W)}:\alpha\in\mathbb{R}_{M}\times\mathcal{H}^{M},\ 0\leq t\leq\tau\}$
are P-Glivenko-Cantelli.\end{lem}
\begin{proof}
Given that $\eta_{\alpha}(W)=\theta'Z+\eta_{\beta}(X)$ is bounded almost
surely, it is easy to see that $\Delta m_{0}(t,W,\alpha)=\Delta[\eta_{\alpha}(W)-\log S_{0}(t,\alpha)]1_{\{0\leq t\leq\tau\}}$
and $Y(t)e^{\eta_{\alpha}(W)}$ are bounded. So following Theorem 19.13
in \cite{Vaart2000}, it is sufficient to show that $\mathcal{N}(\epsilon,\mathcal{F}_{i},L_{1}(P))<\infty$
for $i=1,2$.

For any $f=\Delta m_{0}(t,W,\alpha)$, and $f_{1}=\Delta m_{0}(t,W,\alpha_{1})$
in $\mathcal{F}_{1}$, 
\begin{eqnarray*}
||f-f_{1}||_{L_{1}(P)} & = & P|f-f_{1}|=P|\Delta m_{0}(\cdot,\alpha)-\Delta m_{0}(\cdot,\alpha_{1})|\\
 & = & P|\Delta\big[\eta_{\alpha}(W)-\eta_{\alpha_{1}}(W)-\log\frac{S_{0}(\cdot,\alpha)}{S_{0}(\cdot,\alpha_{1})}\big]1_{\{0\leq T\leq\tau\}}|\\
 & \leq & P|\eta_{\alpha}(W)-\eta_{\alpha_{1}}(W)|+P|[\log S_{0}(\cdot,\alpha)-\log S_{0}(\cdot,\alpha_{1})]1_{\{0\leq T\leq\tau\}}|\\
 & \lesssim & P|\eta_{\alpha}(W)-\eta_{\alpha_{1}}(W)|+\sup_{0\leq t\leq\tau}|S_{0}(t,\alpha)-S_{0}(t,\alpha_{1})|\\
 & \lesssim & P|\eta_{\alpha}(W)-\eta_{\alpha_{1}}(W)|+\sup_{0\leq t\leq\tau}|E(Y(t)e^{\eta_{\alpha}(W)}-Y(t)e^{\eta_{\alpha_{1}}(W)})|\\
 & \lesssim & P|\eta_{\alpha}(W)-\eta_{\alpha_{1}}(W)|.
\end{eqnarray*}
Therefore $\mathcal{N}(\epsilon,\mathcal{F}_{1},L_{1}(P))\asymp\mathcal{N}(\epsilon,\{\eta_{\alpha}(W):\alpha\in\mathbb{R}_{M}\times\mathcal{H}^{M}\},L_{1}(P))$.

Similarly for $f=Y(t)e^{\eta_{\alpha}(W)}$, and $f_{1}=Y(t)e^{\eta_{\alpha_{1}}(W)}:$
in $\mathcal{F}_{2}$, 
\begin{eqnarray*}
||f-f_{1}||_{L_{1}(P)} & = & P|f-f_{1}|\\
 & \leq & P|e^{\eta_{\alpha}(W)}-e^{\eta_{\alpha_{1}}(W)}|\\
 & \lesssim & P|\eta_{\alpha}(W)-\eta_{\alpha_{1}}(W)|,
\end{eqnarray*}
and $\mathcal{N}(\epsilon,\mathcal{F}_{2},L_{1}(P))\asymp\mathcal{N}(\epsilon,\{\eta_{\alpha}(W):\alpha\in\mathbb{R}_{M}\times\mathcal{H}^{M}\},L_{1}(P))$. 

So it suffices to show that $\mathcal{N}(\epsilon,\{\eta_{\alpha}(W):\alpha\in\mathbb{R}_{M}\times\mathcal{H}^{M}\},L_{1}(P))<\infty$,
\textcolor{black}{which is obvious since $\eta_{\alpha}(W)$ is bounded
almost surely for $\alpha\in\mathbb{R}_{M}\times\mathcal{H}^{M}$.}\end{proof}

\begin{lem}
\label{lem: conv I and III}Let $I$ and $III$ be defined as 
\begin{eqnarray*}
I & = & (P_{\Delta n}-P_{\Delta})(m_{0}(\cdot,\alpha)-m_{0}(\cdot,\alpha_{0})),\\
III & = & (P_{n}-P)\big\{ Y(t)\big[\exp(\eta_{\alpha}(W))-\exp(\eta_{\alpha_{0}}(W))\big]\big\},
\end{eqnarray*}
and \textup{$\mathcal{B}_{\delta}=\{\alpha\in\mathbb{R}^{p}\times\mathcal{H}(K):\delta/2\leq d(\alpha,\alpha_{0})\leq\delta\}$,
}then 
\begin{eqnarray*}
\sup_{\alpha\in\mathcal{B}_{\delta}}I & = & O(\delta^{\frac{2r-1}{2r}}n^{-1/2}), \\
\sup_{\alpha\in\mathcal{B}_{\delta}}III & = & O(\delta^{\frac{2r-1}{2r}}n^{-1/2}),\quad \text{   for for } t\in [0,\tau]. 
\end{eqnarray*}
\end{lem}
\begin{proof}
Consider 
\begin{eqnarray*}
\mathcal{M}_{\delta1} & = & \{\Delta[m_{0}(t,W,\alpha)-m_{0}(t,W,\alpha_{0})]1_{\{0\leq t\leq\tau\}},\ \alpha\in\mathcal{B}_{\delta}\},\\
\mathcal{M}_{\delta2} & = & \{Y(t)\big[\exp(\eta_{\alpha}(W))-\exp(\eta_{\alpha}(W))\big],\ \alpha\in\mathcal{B}_{\delta},\ t\in[0,\tau]\},
\end{eqnarray*}
with $L_{2}(P)$ norm, i.e for any $f\in\mathcal{M}_{\delta1}$, $||f||_{P,2}=(\int f^{2}dP)^{1/2}=\big(E^{t,W}f^{2}(t,W,\alpha)\big)^{1/2}$, 
and for any $f\in\mathcal{M}_{\delta2}$, $||f||_{P,2}=(\int f^{2}dP)^{1/2}=\big(E^{T,W}f^{2}(T,W,t,\alpha)\big)^{1/2}$
. Then it suffices to show that

\[
\big|\big|||\mathbb{G}_{n}||_{\mathcal{M}_{\delta1}}\big|\big|_{P,2}=O(\delta^{\frac{2r-1}{2r}}),
\]

\[
\big|\big|||\mathbb{G}_{n}||_{\mathcal{M}_{\delta2}}\big|\big|_{P,2}=O(\delta^{\frac{2r-1}{2r}}),
\]
where $\mathbb{G}_{n}=\sqrt{n}(P_{n}-P)$ and $||\mathbb{G}_{n}||_{\mathcal{M}_{\delta i}}=\sup_{f\in\mathcal{M}_{\delta i}}|\mathbb{G}_{n}f|,\ i=1,2$. 

We first show that 
\[
\log\mathcal{N}(\epsilon,\mathcal{M}_{\delta1},||\cdot||_{p,2})\leq O((p+\epsilon^{-1/r})\log(\frac{\delta}{\epsilon})),
\]
and 
\[
\log\mathcal{N}(\epsilon,\mathcal{M}_{\delta2}(t),||\cdot||_{p,2})\leq O((p+\epsilon^{-1/r})\log(\frac{\delta}{\epsilon})),\ \ \text{for all }t\in[0,\tau].
\]

Suppose there exist functions $f_{1},\ldots,f_{m}\in\mathcal{M}_{\delta1}$,
such that 
\[
\min_{1\leq i\leq m}||f-f_{i}||_{p,2}<\epsilon, \ \ \ \text{for all }f\in\mathcal{M}_{\delta1}.
\]
This is equivalent to the existence of $\alpha_{1},\ldots\alpha_{m}\in\mathcal{B}_{\delta}$,\ 
s.t 
\[
\min_{1\leq i\leq m}\big|\big|\Delta[m_{0}(\cdot,\alpha)-m_{0}(\cdot,\alpha_{i})]1_{\{0\leq T\leq\tau\}}\big|\big|_{p,2}<\epsilon,\ \ \ \text{for all }\alpha\in\mathcal{B}_{\delta}.
\]
Observe that
\begin{eqnarray*}
 &  & \{\Delta[m_{0}(t,W,\alpha)-m_{0}(t,W,\alpha_{i})]1_{\{0\leq t\leq\tau\}}\}^{2}\\
 & = & \Delta\big[\eta_{\alpha}(W)-\eta_{\alpha_{i}}(W)-\log\frac{S_{0}(t,\alpha)}{S_{0}(t,\alpha_{i})}\big]{}^{2}1_{\{0\leq t\leq\tau\}}\\
 & \leq & 2\Delta\{[\eta_{\alpha}(W)-\eta_{\alpha_{i}}(W)]^{2}+[\log\frac{S_{0}(t,\alpha)}{S_{0}(t,\alpha_{i})}]^{2}\}1_{\{0\leq t\leq\tau\}}\\
 & \leq & 2\Delta\{[\eta_{\alpha}(W)-\eta_{\alpha_{i}}(W)]^{2}+c[S_{0}(t,\alpha)-S_{0}(t,\alpha_{i})]^{2}\}1_{\{0\leq t\leq\tau\}}\\
 & = & 2\Delta\big\{[\eta_{\alpha}(W)-\eta_{\alpha_{i}}(W)]^{2}+c[\mathbb{E}Y(t)\{\exp(\eta_{\alpha}(W))-\exp(\eta_{\alpha_{i}}(\alpha))\}]^{2}\big\}1_{\{0\leq t\leq\tau\}}\\
 & \leq & 2\Delta\big\{[\eta_{\alpha}(W)-\eta_{\alpha_{i}}(W)]^{2}+c\mathbb{E}Y^{2}(t)\mathbb{E}\big[\exp(\eta_{\alpha}(W))-\exp(\eta_{\alpha i}(W))\big]{}^{2}\big\}1_{\{0\leq t\leq\tau\}}\\
 & \leq & 2\Delta\big\{[\eta_{\alpha}(W)-\eta_{\alpha_{i}}(W)]^{2}+c_{1}\mathbb{E}Y(t)\mathbb{E}\big[\eta_{\alpha}(W)-\eta_{\alpha_{i}}(W)\big]{}^{2}\big\}1_{\{0\leq t\leq\tau\}}.
\end{eqnarray*}
Then
\begin{align*}
 & \big|\big|\Delta[m_{0}(\cdot,\alpha)-m_{0}(\cdot,\alpha_{i})]1_{\{0\leq T\leq\tau\}}\big|\big|_{p,2}^{2}\\
 & =P\{\Delta[m_{0}(\cdot,\alpha)-m_{0}(\cdot,\alpha_{i})]1_{\{0\leq T\leq\tau\}}\}^{2}\\
 & \lesssim d^{2}(\alpha,\alpha_{i}).
\end{align*}
Therefore, the covering number for $\mathcal{M}_{\delta1}$ is of the
same order as that for $\mathcal{B}_{\delta}$. To be more specific,
\begin{equation}
N(\epsilon,\mathcal{M}_{\delta1},||\cdot||_{p,2})\leq N(\epsilon/C,\mathcal{B}_{\delta},d).\label{eq:covering number 1}
\end{equation}
In addition, we know that 
\[
d^{2}(\alpha,\alpha_{i})\leq2\mathbb{E}\Delta[(\theta-\theta_{i})'Z]^{2}+2d^{2}(\beta,\beta_{i}),
\]
and it follows that $N(\epsilon/C,\mathcal{B}_{\delta},d)\leq N(\epsilon/2C,\mathcal{B}_{\delta}^{\theta},d_{\theta})\cdot N(\epsilon/2C,\mathcal{B}_{\delta}^{\beta},d_{\beta})$,
where $d_{\theta}^{2}(\theta_{1},\theta_{2})=\mathbb{E}\Delta[(\theta_{1}-\theta_{2})'Z]^{2}$
and $d_{\beta}(\beta_{1},\beta_{2})=d(\beta_{1},\beta_{2})$.  Here    $\mathcal{B}_{\delta}^{\theta}$ 
and $\mathcal{B}_{\delta}^{\beta}$  are defined as 
\[
\mathcal{B}_{\delta}^{\theta}=\{\theta\in\mathbb{R}^{p},d_{\theta}(\theta,\theta_{0}))\leq\delta\},\ \ \ \ \mathcal{B}_{\delta}^{\beta}=\{\beta\in\mathcal{H}(K),d_{\beta}(\beta,\beta_{0})\leq\delta\}, 
\]
with \ $\mathcal{B}_{\delta}^{\theta}\times\mathcal{B}_{\delta}^{\beta}\supset\mathcal{B}_{\delta}$.

It is easy to see that $N(\epsilon/2C,\mathcal{B}_{\delta}^{\theta},d_{\theta})=O((\frac{\delta}{\epsilon})^{p})$.
For $N(\epsilon/2C,\mathcal{B}_{\delta}^{\beta},d_{\beta})$, noticing
that $\mathcal{H}(K)=L_{K^{1/2}}({L}_{2})=\{\sum_{k}b_{k}L_{K^{1/2}}\phi_{k}:(b_{k})\in l_{2}\}$,
then for any $\beta=\sum_{k\geq1}b_{k}L_{K^{1/2}}\phi_{k}\in\mathcal{H}(K)$,
we have

\begin{eqnarray*}
d^{2}(\text{\ensuremath{\beta}},\beta_{0}) & = & \mathbb{E}\Delta\eta_{\beta-\beta_{0}}^{2}(X)\\
 & = & <\beta-\beta_{0},L_{C_{\Delta}}\beta-\beta_{0}>_{{L}_{2}}\\
 & = & <\sum_{k\geq1}(b{}_{k}-b_{k}^{0})L_{K^{1/2}}\phi_{k},\sum_{k\geq1}(b_{k}-b_{k}^{0})L_{C_{\Delta}K^{1/2}}\phi_{k}>_{{L}_{2}}\\
 & = & <\sum_{k\geq1}(b_{k}-b_{k}^{0})\phi_{k},\sum_{k\geq1}(b_{k}-b_{k}^{0})L_{K^{1/2}C_{\Delta}K^{1/2}}\phi_{k}>_{{L}_{2}}\\
 & = & <\sum_{k\geq1}(b_{k}-b_{k}^{0})\phi_{k},\sum_{k\geq1}(b_{k}-b_{k}^{0})s_{k}\phi_{k}>_{{L}_{2}}\\
 & = & \sum_{k\geq1}s_{k}(b_{k}-b_{k}^{0})^{2}.
\end{eqnarray*}
If we further let $\gamma_{k}=\sqrt{s_{k}}b_{k}$ , then $d(\beta,\beta_{0})=\sum_{k\geq1}(\gamma_{k}-\gamma_{k}^{0})^{2}$ and
$\mathcal{B}_{\delta}^{\beta}=\{\beta\in\mathcal{H}(K):d(\beta,\beta_{0}))\leq\delta\}$
can be rewritten as 
\[
\mathcal{B}_{\delta}=\{\beta=\sum_{k\geq1}s_{k}^{-1/2}\gamma_{k}L_{K^{1/2}}\phi_{k}:\,(s_{k}^{-1/2}\gamma_{k})\in l_{2},\ \sum_{k\geq1}(\gamma_{k}-\gamma_{k}^{0})^{2}\leq\delta^{2}\}.
\]
Let $M=(\frac{\epsilon}{4C})^{-1/r}$, and 
\[
\mathcal{B}_{\delta}^{\beta*}=\{\beta=\sum_{k=1}^{M}s_{k}^{-1/2}\gamma_{k}L_{K^{1/2}}\phi_{k}:\,(s_{k}^{-1/2}\gamma_{k})_{k=1}^{M}\in l_{2},\ \sum_{k=1}^{M}(\gamma_{k}-\gamma_{k}^{0})^{2}\leq\delta^{2}\}.
\]
For any $\beta=\sum_{k\geq1}s_{k}^{-1/2}\gamma_{k}L_{K^{1/2}}\phi_{k}\in\mathcal{B}_{\delta}$,
let $\beta^{*}=\sum_{k=1}^{M}s_{k}^{-1/2}\gamma_{k}L_{K^{1/2}}\phi_{k}\in\mathcal{B}_{\delta}^{*}$.
It's easy to see that 
\begin{eqnarray*}
d^{2}(\beta,\beta^{*}) & = & \sum_{k>M}\gamma_{k}^{2}=\sum_{k>M}s_{k}b_{k}^{2}\leq s_{M}\sum_{k>M}b_{k}^{2}\asymp M^{-2r}=(\frac{\epsilon}{4C})^{2},
\end{eqnarray*}
where $\sum_{k>M}b_{k}^{2}$ is some small number when $M$ is large, 
since $(b_{k})\in l_{2}$ . So if we can find a set $\{\beta_{i}^{*}\}_{i=1}^{m}\subset\mathcal{B}_{\delta}^{*}$ 
satisfying  
\[
\min_{1\leq k\leq m}d(\beta^{*},\beta_{i}^{*})\leq\epsilon/4C\ \ \text{for all }\beta^{*}\in\mathcal{B}_{\delta}^{*},
\]
then it also guarantees that 
\[
\min_{1\leq k\leq m}d(\beta,\beta_{i}^{*})\leq\min_{1\leq k\leq m}[d(\beta,\beta^{*})+d(\beta^{*},\beta_{i}^{*})]\lesssim\epsilon/2C\ \ \text{for all }\beta\in\mathcal{B}_{\delta},
\]
i.e. 
\begin{equation}
N(\epsilon/2C,\mathcal{B}_{\delta}^{\beta},d_{\beta})\lesssim N(\epsilon/4C,\mathcal{B}_{\delta}^{*},d).\label{eq:covering number 2}
\end{equation}
We know that $N(\epsilon/4C,\mathcal{B}_{\delta}^{*},d)\leq(\frac{4\delta+\epsilon/4C}{\epsilon/4C})^{M}$
is the covering number for a ball in $\mathbb{R}^{M}$. Therefore
combining with  $(\ref{eq:covering number 1})$,  we have 
\begin{align*}
\log\mathcal{N}(\epsilon,\mathcal{M}_{\delta1},||\cdot||_{p,2}) & \leq\log N(\epsilon/C,\mathcal{B}_{\delta},d)\\
 & \leq\log N(\epsilon/2C,\mathcal{B}_{\delta}^{\theta},d_{\theta})+\log N(\epsilon/2C,\mathcal{B}_{\delta}^{\beta},d_{\beta})\\
 & \leq(\frac{\epsilon}{4C})^{-1/r}\log(\frac{4\delta+\epsilon/4C}{\epsilon/4C})+\log O((\frac{\delta}{\epsilon})^{p})\\
 & =O((p+\epsilon^{-1/r})\log(\frac{\delta}{\epsilon})).
\end{align*}

Similarly,
\begin{align*}
 & \big|\big|Y(t)\big[\exp(\eta_{\alpha_{1}}(W))-\exp(\eta_{\alpha_{2}}(W))\big]\big|\big|_{p,2}^{2}\\
 & = P^{TW}\{Y(t)\big[\exp(\eta_{\alpha_{1}}(W))-\exp(\eta_{\alpha_{2}}(W))\big]\}^{2} \\
 & \leq Cd^{2}(\alpha_{1},\alpha_{2}),\quad\text{for all }t\in[0,\tau].
\end{align*}

Following the same procedure, we have 
\[
\log\mathcal{N}(\epsilon,\mathcal{M}_{\delta2},||\cdot||_{p,2})\leq O((p+\epsilon^{-1/r})\log(\frac{\delta}{\epsilon})).
\]

Now we are able to calculate $J(1,\mathcal{M}_{\delta1})$, 
\begin{eqnarray*}
J(1,\mathcal{M}_{\delta1}) & = & \int_{0}^{1}\sqrt{1+\log\mathcal{N}(\epsilon,\mathcal{M}_{\delta1},||\cdot||_{p,2})}d\epsilon\\
 & = & \int_{0}^{1}\sqrt{1+(p+\epsilon^{-1/r})\log(\frac{\delta}{\epsilon})}d\epsilon\\
 & \asymp & \int_{0}^{1}\sqrt{\epsilon^{-1/r}\log(\frac{\delta}{\epsilon})}d\epsilon,\\
\mbox{and \ for }  \  u=\sqrt{\log(\frac{\delta}{\epsilon})},  \   & \asymp & \int_{\sqrt{\log\delta}}^{\infty}(\frac{\delta}{e^{u^{2}}})^{-\frac{1}{2r}}u^{2}\cdot2\delta e^{-u^{2}}du\\
 & = & O(\delta{}^{\frac{2r-1}{2r}})\int_{\sqrt{\log\delta}}^{\infty}(e^{-u^{2}})^{(1-\frac{1}{2r})}u^{2}\cdot du\\
 & = & O(\delta{}^{\frac{2r-1}{2r}}), \ \ \ \ \ \mbox{for} \  r>\frac{1}{2}.
\end{eqnarray*}
The last inequality follows from the fact that the integral above can be seen
as the second order moment of a standard normal times some constant, hence it is a constant not depending on $\delta$. Since
\textcolor{black}{ functions in $\mathcal{M}_{\delta1}$
are bounded} and $J(1,\mathcal{M}_{\delta1})=O(\delta{}^{\frac{2r-1}{2r}})$,
 Theorem 2.14.1 in \cite{Vaart1996} implies  
\[
\big|\big|||\mathbb{G}_{n}||_{\mathcal{M}_{\delta1}}\big|\big|_{P,2}\lesssim J(1,\mathcal{M}_{\delta1})=O(\delta{}^{\frac{2r-1}{2r}}).
\]

Similarly we have 
\[
\big|\big|||\mathbb{G}_{n}||_{\mathcal{M}_{\delta2}(t)}\big|\big|_{P,2}=O(\delta^{\frac{2r-1}{2r}}),\ \ \text{for all }t\in[0,\tau].
\]
\end{proof}

\begin{lem}
$\ $\textup{\label{lm: normal 2}$\ $}
\begin{equation}
P_{\Delta n}\{s_{n}(\cdot,\hat{\alpha})[Z]-s_{n}(\cdot,\alpha_{0})[Z]\}-P_{\Delta}\{s(\cdot,\hat{\alpha})[Z]-s(\cdot,\alpha_{0})[Z]\}=o_{p}(n^{-1/2})\label{eq: AN th lem11},
\end{equation}
\begin{equation}
P_{\Delta n}\{s_{n}(\cdot,\hat{\alpha})[g^{*}]-s_{n}(\cdot,\alpha_{0})[g^{*}]\}-P_{\Delta}\{s(\cdot,\hat{\alpha})[g^{*}]-s(\cdot,\alpha_{0})[g^{*}]\}=o_{p}(n^{-1/2}).  \label{eq: AN th lem12}
\end{equation}
\end{lem}
\begin{proof}
We only prove ($\ref{eq: AN th lem12}$) as  the proof of ($\ref{eq: AN th lem11}$)
is similar. The right-hand side of ($\ref{eq: AN th lem12}$) can
be bounded by the sum of the following two terms
\[
I_{1n}=\big|(P_{\Delta n}-P_{\Delta})\{s(\cdot,\hat{\alpha})[g^{*}]-s(\cdot,\alpha_{0})[g^{*}]\}\big|,
\]
and 
\[
I_{2n}=\big|P_{\Delta n}\{s_{n}(\cdot,\hat{\alpha})[g^{*}]-s_{n}(\cdot,\alpha_{0})[g^{*}]-s(\cdot,\hat{\alpha})[g^{*}]+s(\cdot,\alpha_{0})[g^{*}]\}\big|.
\]
We are going to show that $I_{1n}=o_{p}(n^{-\frac{1}{2}})$ and $I_{2n}=o_{p}(n^{-\frac{1}{2}})$.

For the first term, since $S_{0}(\cdot,\hat{\alpha})$, $S_{0}(\cdot,\alpha_{0})$
and $S_{1}(t,\alpha_{0})[g^{*}]$ are bounded almost surely,  we have 
\begin{align*}
I_{1n} & =\big|(P_{n}-P)\{\Delta[\frac{S_{1}(\cdot,\hat{\alpha})[g^{*}]}{S_{0}(\cdot,\hat{\alpha})}-\frac{S_{1}(\cdot,\alpha_{0})[g^{*}]}{S_{0}(\cdot,\alpha_{0})}]\}\big|\\
 & =\big|(P_{n}-P)\{\Delta[S_{0}(\cdot,\hat{\alpha})]^{-1}\big[S_{1}(\cdot,\hat{\alpha})[g^{*}]-S_{1}(\cdot,\alpha_{0})[g^{*}]\big]\\
 & \ \ \ \ \ \ \ \ \ +\Delta[S_{0}(\cdot,\hat{\alpha})S_{0}(\cdot,\alpha_{0})]^{-1}S_{1}(\cdot,\alpha_{0})[g^{*}][S_{0}(\cdot,\hat{\alpha})-S_{0}(\cdot,\alpha_{0})]\}\big|\\
 & \lesssim\big|(P_{n}-P)\{\Delta\big[S_{1}(\cdot,\hat{\alpha})[g^{*}]-S_{1}(\cdot,\alpha_{0})[g^{*}]\big]\}\big|\\
 & \ \ \ \ \ \ \ \ \ +\big|(P_{n}-P)\{\Delta[S_{0}(\cdot,\hat{\alpha})-S_{0}(\cdot,\alpha_{0})]\}\big|.
\end{align*}
Considering $\mathcal{M}_{\delta3}=\big\{\Delta\big[S_{1}(t,\alpha)[g^{*}]-S_{1}(t,\alpha_{0})[g^{*}]\big],\ \alpha\in\mathcal{B}_{\delta}\big\}$,
for any $f_{1},f_{2}\in\mathcal{M}_{\delta1}$, we have
\begin{align*}
||f_{1}-f_{2}||_{p,2} & =\mathbb{E}\Delta^{2}\{S_{1}(\cdot,\alpha_{1})[g^{*}]-S_{1}(\cdot,\alpha_{2})[g^{*}]\}^{2}\\
 & =\mathbb{E}^{\Delta,t,X}\Delta\{\mathbb{E}Y(t)(e^{\eta_{\alpha_{1}}(W)}-e^{\eta_{\alpha_{2}}(W)})\eta_{g^{*}}(X)\}^{2}\\
 & \lesssim d^{2}(\alpha_{1},\alpha_{2}).
\end{align*}
Following the same proof as Lemma $\ref{lem: conv I and III}$, we
can show that 
\[
\big|(P_{n}-P)\{\Delta\big[S_{1}(\cdot,\hat{\alpha})[g^{*}]-S_{1}(\cdot,\alpha_{0})[g^{*}]\big]\}\big|=O(d^{\frac{2r-1}{2r}}(\hat{\alpha},\alpha_{0})n^{-\frac{1}{2}})=o_{p}(n^{-\frac{1}{2}}),
\]
given that $d(\hat{\alpha},\alpha_{0})=O_{p}(n^{-\frac{2r}{2r+1}})$. Similarly,
\[
\big|(P_{n}-P)\{\Delta[S_{0}(\cdot,\hat{\alpha})-S_{0}(\cdot,\alpha_{0})]\}\big|=o_{p}(n^{-\frac{1}{2}}),
\]
and altogether we have shown that $I_{1n}=o_{p}(n^{-\frac{1}{2}})$.

For the second term,    the quantity  inside the empirical measure $P_{\Delta n}$
is 
\[
II_{2n}(t):=\frac{S_{1n}(t,\hat{\alpha})[g^{*}]}{S_{0n}(t,\hat{\alpha})}-\frac{S_{1n}(t,\alpha_{0})[g^{*}]}{S_{0n}(t,\alpha_{0})}-\frac{S_{1}(t,\hat{\alpha})[g^{*}]}{S_{0}(t,\hat{\alpha})}+\frac{S_{1}(t,\alpha_{0})[g^{*}]}{S_{0}(t,\alpha_{0})}.
\]
It follows from the same proof as in Lemma A.7 of  \cite{Huang1999} that 
\[
\sup_{0\leq t\leq1}|II_{2n}(t)|=o_{p}(n^{-\frac{1}{2}}).
\]
\end{proof}

\begin{lem}
$\ $\textup{\label{lm: normal 3}$\ $}
\begin{flalign*}
 & P_{\Delta}\{s(\cdot,\hat{\alpha})[Z-g^{*}]-s(\cdot,\alpha_{0})[Z-g^{*}]\}\\
 & =P_{\Delta}\{s(\cdot,\alpha_{0})[Z-g^{*}]\}^{\otimes2}(\hat{\theta}-\theta_{0})+O(||\hat{\theta}-\theta_{0}||^{2}+||\hat{\beta}-\beta)||_{C_{\Delta}}^{2})\\
 & =P_{\Delta}\{s(\cdot,\alpha_{0})[Z-g^{*}]\}^{\otimes2}(\hat{\theta}-\theta_{0})+o_{p}(n^{-1/2}).
\end{flalign*}
\end{lem}
\begin{proof}
By lemma $\ref{lem: dH}$, direct calculation implies 
\begin{align*}
 & P_{\Delta}\{s(\cdot,\hat{\alpha})[Z-g^{*}]-s(\cdot,g_{0})[Z-h^{*}]\}\\
= & P_{\Delta}\{s(\cdot,\alpha_{0})[Z-g^{*}]s(\cdot,g_{0})[\hat{\alpha}-\alpha_{0}]\}+O(d^{2}(\hat{\alpha},\alpha_{0}))\\
= & P_{\Delta}\{s(\cdot,\alpha_{0})[Z-g^{*}]s(\cdot,g_{0})[Z]\}(\hat{\theta}-\theta_{0})\\
 & \ \ \ +P_{\Delta}\{s(\cdot,\alpha_{0})[Z-g^{*}]s(\cdot,g_{0})[\eta_{\hat{\beta}}-\eta_{\beta_{0}}]\}\\
 & \ \ \ +O(d^{2}(\hat{\alpha},\alpha_{0})),
\end{align*}
while by ($\ref{eq: a*}$) , ($\ref{eq:g*}$) and ($\ref{eq:ssieni}$),
we have 
\begin{align*}
 & P_{\Delta}\{s(\cdot,\alpha_{0})[Z-g^{*}]s(\cdot,g_{0})[\eta_{\hat{\beta}}-\eta_{\beta_{0}}]\}\\
 & =P_{\Delta}[Z-\eta_{g^{*}}(X)-\frac{S_{1}(t,\alpha_{0})[Z-g^{*}]}{S_{0}(t,\alpha_{0})}][\eta_{\hat{\beta}}-\eta_{\beta_{0}}-\frac{S_{1}(t,\alpha_{0})[\hat{\beta}-\beta_{0}]}{S_{0}(t,\alpha_{0})}]\\
 & =P_{\Delta}\{Z-\eta_{g^{*}}(X)-\mathbb{E}[Z-\eta_{g^{*}}(X)|T,\Delta=1]\}\{\eta_{\hat{\beta}-\beta_{0}}(X)-\mathbb{E}[\eta_{\hat{\beta}-\beta_{0}}(X)|T,\Delta=1]\}\\
 & =P_{\Delta}[Z-\eta_{g^{*}}(X)-a^{*}(T)][\eta_{\hat{\beta}-\beta_{0}}(X)-a(T)]\\
 & =0,
\end{align*}
and 
\[
P_{\Delta}\{s(\cdot,\alpha_{0})[Z-g^{*}]s(\cdot,g_{0})[Z]\}=P_{\Delta}\{s(\cdot,\alpha_{0})[Z-g^{*}]\}^{\otimes2}.
\]
The lemma now follows from from  Theorem $\ref{thm: convergence rate}$ which asserts that  $d^{2}(\hat{\alpha},\alpha_{0})=o_{p}(n^{-1/2})$. 
\end{proof}

\bibliography{CoxReference}
\end{document}